\documentclass{lmcs} 
\pdfoutput=1

\usepackage{lastpage}
\lmcsdoi{15}{1}{13}
\lmcsheading{}{\pageref{LastPage}}{}{}%
{Dec.~12,~2017}{Feb.~13,~2019}{}

\keywords{quantified Boolean formulas, proof complexity, lower bounds.\\ \indent\textit{ACM subject classification:} F.2.2 Nonnumerical Algorithms and Problems: Complexity of proof procedures}

\usepackage{amsmath}
\usepackage{hyperref}
\usepackage{tikz}
\usetikzlibrary{arrows}
\usepackage{amssymb}
\usepackage{xspace}
\usepackage{thm-restate}
\usepackage{framed}
\usepackage{bussproofs}
\usepackage{bm}

\let\Phi\varPhi
\let\Psi\varPsi

\newcommand{\proofSystemFont}[1]{\mathsf{#1}}
\newcommand{\PS}{\proofSystemFont{P}}
\newcommand{\Pred}{\proofSystemFont{P{\mbox{\upshape{+}}}{\forall}red}}
\newcommand{\SigmaPred}{\proofSystemFont{\Sigma_1^p{\mbox{\upshape{-}}}P{\mbox{\upshape{+}}}{\forall}red}}
\newcommand{\red}{\proofSystemFont{{\forall}red}}
\newcommand{\Frege}{\proofSystemFont{Frege}}
\newcommand{\CP}{\proofSystemFont{CP}}
\newcommand{\PCR}{\proofSystemFont{PCR}}
\newcommand{\SOS}{\proofSystemFont{SOS}}

\newcommand{\QRes}{\proofSystemFont{Q{\mbox{\upshape{-}}}Res}}
\newcommand{\QURes}{\proofSystemFont{QU{\mbox{\upshape{-}}}Res}}
\newcommand{\LDQRes}{\proofSystemFont{LD{\mbox{\upshape{-}}}Q{\mbox{\upshape{-}}}Res}}
\newcommand{\IRC}{\proofSystemFont{IR{\mbox{\upshape{-}}}calc}}
\newcommand{\CPred}{\proofSystemFont{CP{\mbox{\upshape{+}}}{\forall}red}}
\newcommand{\PCRred}{\proofSystemFont{PCR{\mbox{\upshape{+}}}{\forall}red}}
\newcommand{\Fred}{\proofSystemFont{Frege{\mbox{\upshape{+}}}{\forall}red}}
\newcommand{\SOSred}{\proofSystemFont{SOS{\mbox{\upshape{+}}}{\forall}red}}

\newcommand{\complexityClassFont}[1]{\mathsf{#1}}

\newcommand{\NP}{\complexityClassFont{NP}}
\newcommand{\PSPACE}{\complexityClassFont{PSPACE}}
\newcommand{\AC}[1]{\complexityClassFont{AC^{#1}}}

\newcommand{\Q}{\ensuremath{\mathcal{Q}}\xspace}
\newcommand{\R}{\ensuremath{\mathcal{R}}\xspace}
\newcommand{\A}{\ensuremath{\mathcal{A}}\xspace}
\newcommand{\St}{\ensuremath{S}}
\newcommand{\QBF}[2]{\ensuremath{{#1}\cdot{#2}}\xspace}

\newcommand{\var}{\mbox{\upshape var}}
\newcommand{\vars}{\mbox{\upshape vars}}

\newcommand{\rng}{\mbox{\upshape rng}}

\newcommand{\dom}{\mbox{\upshape dom}}
\newcommand{\capa}{\mbox{\upshape capacity}}
\newcommand{\cost}{\mbox{\upshape cost}}
\newcommand{\proj}{\mbox{\upshape proj}}
\newcommand{\sgn}{\mbox{\upshape sgn}}

\newcommand{\eq}{EQ}

\newcommand{\JLBFdiamond}{%
\begin{tikzpicture}[>=triangle 45, xscale=1.6,yscale=2.1]

\node[draw, align=center, minimum width=5em, minimum height = 3em, rounded corners=0.15cm](N) at (0,0){$\Fred$};
\node[draw, align=center, minimum width=5em, minimum height = 3em, rounded corners=0.15cm](W) at (-1,-1){$\CPred$};
\node[draw, align=center, minimum width=5em, minimum height = 3em, rounded corners=0.15cm](E) at (1,-1){$\PCRred$};
\node[draw, align=center, minimum width=5em, minimum height = 3em, rounded corners=0.15cm](S) at (0,-2){$\QURes$};

\draw [->](N) to (-0.55,-0.55);
\draw (-0.55,-0.55) -- (W);
\draw [->](N) -- (0.55,-0.55);
\draw (0.55,-0.55) -- (E);
\draw [->](W) -- (-0.5,-1.5);
\draw (-0.5,-1.5) -- (S);
\draw [->](E) -- (0.5,-1.5);
\draw (0.5,-1.5) -- (S);
\draw[dashed] (W) -- (E);

\node[draw, align=center, rounded corners=0.05cm](Aii) at (2.8,-0.6){$\mathsf{A}$};
\node[draw, align=center, rounded corners=0.05cm](Bii) at (3.8,-0.6){$\mathsf{B}$};
\draw[dashed](Aii)--(Bii);
\node[align = left, text width = 4cm] at (5.5,-0.6){$\mathsf{A}$ and $\mathsf{B}$ are\\incomparable};

\node[draw, align=center, rounded corners=0.05cm](Ai) at (2.8,-1.4){$\mathsf{A}$};
\node[draw, align=center, rounded corners=0.05cm](Bi) at (3.8,-1.4){$\mathsf{B}$};
\draw[->](Ai)--(3.4,-1.4);
\draw(3.4,-1.4)--(Bi);
\node[align = left, text width = 4cm] at (5.5,-1.4){$\mathsf{A}$ $p$-simulates and \\is exponentially \\separated from $\mathsf{B}$};

\node at (-3,0){};
\end{tikzpicture}
}

\newcommand{\JLBFigUniversalReduction}
{%
\begin{framed}

    \begin{minipage}{.37\linewidth}
      \begin{prooftree}
	\AxiomC{$L$}
	\UnaryInfC{$L[\beta]$}
      \end{prooftree}
    \end{minipage}
    \begin{minipage}{.58\linewidth}
      \begin{itemize}
        \vspace{1em}
        \item $\beta$ is a partial assignment to a universal block $U$ of $\Q$.
        \item $\vars(L)$ contains no variables right of $U$, with respect to $\Q$.\\
      \end{itemize}
    \end{minipage}

\end{framed}
}

\newcommand{\QUResDerivationRules}
{%
\begin{framed}

    \begin{minipage}{.22\linewidth}
    \textbf{Axiom:}
    \end{minipage}
    \begin{minipage}{.33\linewidth}
      \begin{prooftree}
	\AxiomC{}
	\UnaryInfC{$C$}
      \end{prooftree}
    \end{minipage}
    \begin{minipage}{.40\linewidth}
      $C$ is a clause in the matrix $\phi$.
    \end{minipage}
    \vspace{0.5em}\\

    \begin{minipage}{.22\linewidth}
    \textbf{Weakening:}
    \end{minipage}
    \begin{minipage}{.33\linewidth}
      \begin{prooftree}
	\AxiomC{$C$}
	\UnaryInfC{$C \cup W$}
      \end{prooftree}
    \end{minipage}
    \begin{minipage}{.40\linewidth}
      Each variable appearing in $W$ is in $\vars(\Phi)$.\\ 
      The consequent $C \cup W$ is non-tautologous.
    \end{minipage}
    \vspace{0.5em}\\

    \begin{minipage}{.22\linewidth}
    \textbf{Resolution:}
    \end{minipage}
    \begin{minipage}{.33\linewidth}
      \begin{prooftree}
	\AxiomC{$C_1 \cup \{x\}$}
	\AxiomC{$C_2 \cup \{\lnot x\}$}
	\BinaryInfC{$C_1 \cup C_2$}
      \end{prooftree}
    \end{minipage}
    \begin{minipage}{.40\linewidth}
      The resolvent $C_1 \cup C_2$ is non-tautologous.
    \end{minipage}
    \vspace{0.5em}\\

    \begin{minipage}{.22\linewidth}
    \textbf{Universal reduction:}
    \end{minipage}
    \begin{minipage}{.33\linewidth}
      \begin{prooftree}
	\AxiomC{$C \cup U$}
	\UnaryInfC{$C$}
      \end{prooftree}
    \end{minipage}
    \begin{minipage}{.40\linewidth}
      $U$ contains only universal literals.\\
      Each variable in $U$ is right of each existential variable in $C$, with respect to $\Q$.
    \end{minipage}
    \vspace{0.5em}\\

\end{framed}
}

\newcommand{\CPDerivationRules}
{%
\begin{framed}
    \begin{minipage}{.25\linewidth}
    \textbf{Clause Axiom:}
    \end{minipage}
    \begin{minipage}{.4\linewidth}
    \begin{prooftree}
    \AxiomC{}
    \UnaryInfC{$\sum_{l \in C} R(l) \geq 1$}
    \end{prooftree}
    \end{minipage}
    \begin{minipage}{.3\linewidth}
    for any clause $C \in \phi$, where $R(x) = x$, $R(\neg x) = 1-x$
    \end{minipage}
    \vspace{0.5em}\\
    
    \begin{minipage}{.25\linewidth}
    \textbf{Boolean Axiom:}
    \end{minipage}
    \begin{minipage}{.4\linewidth}
    \begin{minipage}{.45\linewidth}
    \begin{prooftree}
    \AxiomC{}
    \UnaryInfC{$x \geq 0$}
    \end{prooftree}
    \end{minipage}
    \begin{minipage}{.45\linewidth}
    \begin{prooftree}
    \AxiomC{}
    \UnaryInfC{$-x \geq -1$}
    \end{prooftree}
    \end{minipage}
    \end{minipage}
    \begin{minipage}{.3\linewidth}
    for any variable $x$
    \end{minipage}
    \vspace{0.5em}\\

    \begin{minipage}{.25\linewidth}
    \textbf{Linear combination:}
    \end{minipage}
    \begin{minipage}{.4\linewidth}
      \begin{prooftree}
	\AxiomC{$\sum_i a_i x_i \geq A$}
	\AxiomC{$\sum_i b_i x_i \geq B$}
	\BinaryInfC{$\sum_i (\alpha a_i + \beta b_i) x_i \geq \alpha A + \beta B$}
      \end{prooftree}
    \end{minipage}
    \begin{minipage}{.3\linewidth}
      for any $\alpha, \beta \in \mathbb{N}$
    \end{minipage}
    \vspace{0.5em}\\
    
    \begin{minipage}{.25\linewidth}
    \textbf{Division:}
    \end{minipage}
    \begin{minipage}{.4\linewidth}
      \begin{prooftree}
	\AxiomC{$\sum_i c a_i x_i \geq A$}
	\UnaryInfC{$\sum_i a_i x_i \geq \left\lceil \frac{A}{c} \right\rceil$}
      \end{prooftree}
    \end{minipage}
    \begin{minipage}{.3\linewidth}
      for any non-zero $c \in \mathbb{N}$
    \end{minipage}
\end{framed}
}

\newcommand{\PCRDerivationRules}
{%
\begin{framed}

    \begin{minipage}{.25\linewidth}
    \textbf{Clause Axiom:}
    \end{minipage}
    \begin{minipage}{.4\linewidth}
      \begin{prooftree}
	\AxiomC{}
	\UnaryInfC{$\prod_{l \in C} V(l) = 0$}
      \end{prooftree}
    \end{minipage}
    \begin{minipage}{.3\linewidth}
      for any clause $C \in \phi$, where $V(x) = x$, $V(\neg x) = \bar{x}$
    \end{minipage}
    \vspace{0.5em}\\
    
     \begin{minipage}{.25\linewidth}
    \textbf{Boolean Axiom:}
    \end{minipage}
    \begin{minipage}{.4\linewidth}
    \begin{minipage}{.45\linewidth}
      \begin{prooftree}
	\AxiomC{}
	\UnaryInfC{$y^2 - y = 0$}
      \end{prooftree}
      \end{minipage}
      \begin{minipage}{.45\linewidth}
      \begin{prooftree}
      \AxiomC{}
      \UnaryInfC{$x + \bar{x} - 1 = 0$}
      \end{prooftree}
      \end{minipage}
    \end{minipage}
    \begin{minipage}{.3\linewidth}
     for any $y \in \{ x, \bar{x} : x \in \vars(\phi) \}$\\
     for any $x \in \vars(\phi)$
    \end{minipage}
    \vspace{0.5em}\\
    
    \begin{minipage}{.25\linewidth}
    \textbf{Linear combination:}
    \end{minipage}
    \begin{minipage}{.4\linewidth}
      \begin{prooftree}
	\AxiomC{$p(\vec{x}) = 0$}
	\AxiomC{$q(\vec{x}) = 0$}
	\BinaryInfC{$\alpha \cdot p(\vec{x}) + \beta \cdot q(\vec{x}) = 0$}
      \end{prooftree}
    \end{minipage}
    \begin{minipage}{.3\linewidth}
      for any $\alpha, \beta \in \mathbb{F}$
    \end{minipage}
    \vspace{0.5em}\\
    
    \begin{minipage}{.25\linewidth}
    \textbf{Multiplication:}
    \end{minipage}
    \begin{minipage}{.4\linewidth}
      \begin{prooftree}
	\AxiomC{$p(\vec{x}) = 0$}
	\UnaryInfC{$y \cdot p(\vec{x}) = 0$}
      \end{prooftree}
    \end{minipage}
    \begin{minipage}{.3\linewidth}
      for any $y \in \{ x , \bar{x} : x \in \vars(\phi) \}$
    \end{minipage}
\end{framed}
}

\begin{document}

\title[Size, Cost, and Capacity]{Size, Cost, and Capacity:\\A Semantic Technique for Hard Random QBFs\rsuper*}
\titlecomment{{\lsuper*}An extended abstract of this paper appeared in the proceedings of ITCS 2018 \cite{BBH18}.}

\author[O.~Beyersdorff]{Olaf Beyersdorff\rsuper{a}}
\author[J.~Blinkhorn]{Joshua Blinkhorn\rsuper{b}}
\author[L.~Hinde]{Luke Hinde\rsuper{c}}
\address{\lsuper{a}Friedrich Schiller University Jena, Germany}
\email{olaf.beyersdorff@uni-jena.de}
\address{\lsuper{b,c}School of Computing, University of Leeds, UK}
\email{\{scjlb,sclpeh\}@leeds.ac.uk}


\newcommand{\josh}[1]{\textcolor{red}{Josh: #1}}
\newcommand{\luke}[1]{\textcolor{blue}{Luke: #1}}


\begin{abstract}
As a natural extension of the SAT problem, an array of proof systems for quantified Boolean formulas (QBF) have been proposed, many of which extend a propositional proof system to handle universal quantification. 
By formalising the construction of the QBF proof system obtained from a propositional proof system by adding universal reduction (Beyersdorff, Bonacina \& Chew, ITCS `16), we present a new technique for proving proof-size lower bounds in these systems.
The technique relies only on two semantic measures: the \emph{cost} of a QBF, and the \emph{capacity} of a proof.
By examining the capacity of proofs in several QBF systems, we are able to use the technique to obtain lower bounds based on cost alone.

As applications of the technique, we first prove exponential lower bounds for a new family of simple QBFs representing equality.
The main application is in proving exponential lower bounds with high probability for a class of randomly generated QBFs, the first `genuine' lower bounds of this kind, which apply to the QBF analogues of resolution, Cutting Planes, and Polynomial Calculus.
Finally, we employ the technique to give a simple proof of hardness for the prominent formulas of Kleine B\"{u}ning, Karpinski and Fl\"{o}gel.
\end{abstract}

\maketitle

\section{Introduction}   %
\label{sec:introduction} %

\subsection{Proof complexity and solving} %
\label{subsec:proof-solving}              %
The central question in \emph{proof complexity} can be stated as follows: Given a logical theory and a provable theorem, what is the size of the shortest proof? 
This question bears tight connections to central problems in computational complexity \cite{Buss12,CookR79} and bounded arithmetic \cite{Krajicek95,CookN10}.

Proof complexity is intrinsically linked to recent noteworthy innovations in solving, owing to the fact that any decision procedure implicitly defines a \emph{proof system} for the underlying language.
Relating the two fields in this way is illuminating for the practitioner; proof-size and proof-space lower bounds correspond directly to best-case running time and memory consumption for the corresponding solver.
Indeed, proof complexity theory has become the main driver for the asymptotic comparison of solving implementations.
However, in line with neighbouring fields (such as computational complexity), it is the central task of demonstrating lower bounds, and of \emph{developing general methods} for showing such results, that proves most challenging for theoreticians.

The desire for general techniques derives from the strength of modern implementations.
Cutting-edge advances in solving, spearheaded by progress in Boolean satisfiability (SAT), appear to provide a means for the efficient solution of computationally hard problems \cite{Vardi14}.
Contemporary SAT solvers routinely dispatch instances in millions of clauses \cite{MalikZ09}, and are effectively employed as $\NP$ oracles in more complex settings \cite{MeelVCFSFIM16}.
The state-of-the-art procedure is based on a propositional proof system called \emph{resolution}, operating on \emph{conjunctive normal form} (CNF) instances using a technique known as \emph{conflict-driven clause learning} (CDCL) \cite{SilvaS96}.
Besides furthering the intense study of resolution and its fragments \cite{Buss12}, the evident success has inevitably pushed research frontiers beyond the $\NP$-completeness of Boolean satisfiability.

\subsubsection*{Beyond propositional satisfiability}
A case in point is the logic of \emph{quantified Boolean formulas} (QBF), a theoretically important class that forms the prototypical $\PSPACE$-complete language \cite{StockmeyerM73}.
QBF extends propositional logic with existential and universal quantification, and consequently offers succinct encodings of concrete problems from conformant planning \cite{Rintanen07,EglyKLP17,CashmoreFG13}, ontological reasoning \cite{Kontchakov09}, and formal verification \cite{Benedetti08}, amongst other areas \cite{Dershowitz05,BloemKS14,StaberB07}.
There is a large body of work on practical QBF solving, and the relative complexities of the associated resolution-type proof systems are well understood \cite{BalabanovWJ14,BeyersdorffCJ15,JanotaM15}.

The semantics of QBF has a neat interpretation as a two-player \emph{evaluation game}.
Given a QBF $\QBF{\Q}{\phi}$, the $\exists$- and $\forall$-players take turns to assign the existential and universal variables of the formula following the order of the quantifier prefix $\Q$.
When all variables are assigned, the $\exists$-player wins if the propositional formula $\phi$ is satisfied; otherwise, the $\forall$-player takes the win. 
A folklore result states that a QBF is false if and only if the $\forall$-player can win the evaluation game by force; that is, if and only if there exists a winning strategy for the universal player.
The concept of \emph{strategy extraction} originates from QBF solving \cite{GoultiaevaGB11}, whereby a winning strategy `extracted' from the proof certifies the truth or falsity of the instance.
In practice it is not merely the truth value of the QBF that is required -- for real-world applications, certificates provide further useful information \cite{StaberB07}. 

A major paradigm in QBF practice is \emph{quantified conflict-driven clause learning} (QCDCL) \cite{GiunchigliaMN09}, a natural extension of CDCL.
The vast majority of QBF solvers build upon existing SAT techniques in a similar fashion. 
Such a notion can hardly be surprising when one considers that an existentially quantified QBF is merely a propositional formula.
The novel challenge for the QBF practitioner, therefore, and the real test of a solver's strength, is in the handling of universal quantification.

Proof-theoretic analysis of associated QBF proof systems makes this notion abundantly clear.
Consider \emph{QU-Resolution} ($\QURes$) \cite{BuningKF95,Gelder12}, a well-studied QBF proof system closely related to QCDCL solving.\footnote{The calculus $\QURes$, proposed by Van Gelder in \cite{Gelder12}, generalises $\QRes$, introduced by Kleine B\"{u}ning et al. in \cite{BuningKF95}, by allowing resolution over universally quantified pivots.}
That calculus simply extends propositional resolution with a \emph{universal reduction} rule, which allows universal literals to be deleted from clauses under certain conditions.
On existentially quantified QBFs, therefore, $\QURes$ is identical to resolution, and proof-size lower bounds for the latter lift immediately to the former.
From the viewpoint of quantified logic, lower bounds obtained in this way are rightly considered \emph{non-genuine}; they belong in the realm of propositional proof complexity, and tell us nothing about the relative strengths of resolution-based QBF solvers.  

Universal reduction is applicable to many suitable propositional proof systems $\PS$, giving rise to a general model for QBF systems in the shape of $\Pred$ \cite{BeyersdorffBC16}, which adds to the propositional rules of $\PS$ the universal reduction rule `$\red$'. 
As a consequence, the phenomenon of genuineness extends well beyond resolution.
In this paper, in addition to resolution we consider three stronger systems:
Cutting Planes ($\CP$), a well-studied calculus that works with linear inequalities;
the algebraic system Polynomial Calculus (with Resolution, $\PCR$);
and Frege's eponymous `textbook' system for propositional logic.
Their simulation order is depicted in Figure~\ref{fig:proof-systems}.
\begin{figure}[t]
\centering
\JLBFdiamond
\caption{The simulation order of the four QBF proof systems featured in this paper. A proof system $\mathsf{A}$ $p$-simulates the system $\mathsf{B}$ if each $\mathsf{B}$-proof of a formula $\Phi$ can be translated in polynomial time into an $\mathsf{A}$-proof of $\Phi$ \cite{CookR79}. If neither $\mathsf{A}$ nor $\mathsf{B}$ p-simulates the other, then they are incomparable. \label{fig:proof-systems}}
\end{figure}

What is generally desired (and seemingly elusive) in the QBF community is the development of \emph{general} techniques for \emph{genuine} lower bounds.\footnote{Instead of using a more descriptive name such as  `quantifier-alternation-exploiting lower bounds' we refer throughout  this work simply to `genuine lower bounds' to preserve readability, hopefully without loss in clarity.}
The current work embraces maximal generality, and contributes a new technique for genuine QBF lower bounds in the general setting of $\Pred$.

\subsubsection*{When is a lower bound genuine?}
Naturally, the aforementioned objections to non-genuine QBF lower bounds may be raised in the abstract setting of $\Pred$, as that system encompasses the propositional proof system $\PS$.
Indeed, given any unsatisfiable propositional formulas that require large proofs in $\PS$, one can easily construct any number of contrived QBF families -- even with arbitrarily many quantifier alternations -- each of which require large proofs in $\Pred$, but whose hardness stems from the original propositional formulas.
That such lower bounds ought to be identified as non-genuine was highlighted in \cite{Chen17} (cf.\ also \cite{BeyersdorffHP17}).

A formal model for genuine QBF lower bounds was proposed in \cite{BeyersdorffHP17}. 
The model essentially adds an $\NP$ oracle to the proof system, which allows any logically correct propositional inference to be derived in a single step.
Propositional hardness is thus removed from the system.
In the context of $\Pred$, the addition of an $\NP$ oracle essentially gives $\PS$-derivations for free; any formula derivable in $\PS$ can be introduced immediately.
If a lower bound persists despite the oracle, then the proof either requires many universal reduction steps, or requires reductions on large lines.

The model provides an adequate definition of a genuine lower bound, and we adopt a similar approach in this article.
We work with \emph{semantic} $\Pred$ refutations, in which any logically correct $\PS$ inference or universal reduction step is allowed (Definition~\ref{def:semantic-ref}).
The technique we introduce gives a lower bound on the size of a semantic $\Pred$ refutation, i.e. a lower bound that still holds with an $\NP$ oracle.
We therefore deal exclusively in genuine results.

\subsubsection*{Random formulas}
In the design and testing of solvers, large sets of formulas are needed to make effective comparisons between implementations. While many formulas have been constructed by hand, often representing some combinatorial principle, it is of clear benefit to have a procedure to randomly generate such formulas. The search for a better understanding of when such formulas are likely to be true or false, and their likely hardness for solvers, brings us to the study of the proof complexity  of random CNFs and QBFs.

In propositional proof complexity, random 3-SAT instances, the most commonly studied random CNFs, are relatively well understood. There is a constant $r$ such that if a random CNF on $n$ variables contains more than $rn$ clauses, then the CNF is unsatisfiable with probability approaching 1 \cite{FrancoP83}; the upper bound for $r$ has regularly been improved (see \cite{DiazKMP08}, and references therein for previous upper bounds).
Further, if the number of clauses is below $n^{6/5 - \epsilon}$, the CNF requires exponential-size resolution refutations with high probability \cite{BeameP96}. Hardness results for random CNFs are also known for Polynomial Calculus \cite{AlekhnovichR01,Ben-SassonI10} and for Cutting Planes \cite{HrubesP17,FlemingPPR17}.

In contrast, comparatively little is known about randomly generated QBFs. The addition of universally quantified variables raises questions as to what model should be used to generate such QBFs -- care is needed to ensure a suitable balance between universal and existential variables.\footnote{If any clause contains only universal variables, then there is a constant-size refutation using only this clause.}
The best-studied model is that of (1,2)-QCNFs \cite{ChenI05}, for which bounds on the threshold number of clauses needed for a false QBF were shown in \cite{CreignouDER15}. However, to the best of our knowledge, nothing has yet been shown on the proof complexity of randomly generated QBFs. Proving such lower bounds constitutes the major application of our new technique.

\subsection{Our contributions} %
\label{sec:contributions}      %
The primary contribution of this work is the proposal of a \emph{novel and semantically-grounded technique} for proving genuine QBF lower bounds in $\Pred$, representing a significant forward step in the understanding of reasons for hardness in the proof complexity of quantified Boolean formulas.

We exemplify the technique with a new family of hard QBFs, notable for their simplicity, which we strongly suggest should henceforth occupy a prominent place in QBF proof complexity.
As our principal application we prove exponential lower-bounds in three concrete $\Pred$ systems for a large class of randomly generated QBFs.
This is the first time that genuine lower bounds have been shown \emph{en masse} for randomly generated QBFs.
Lastly, we note that our technique can be applied to give a simple proof of hardness for a family of well known QBFs from \cite{BuningKF95}.

In addition, we also determine exact conditions on a so-called \emph{base system} $\PS$ by which $\Pred$ is properly defined and receptive to our method.
We detail our contributions below, beginning with a description of the framework, followed by the lower bound technique,  and concluding with several applications. 

\subsubsection{The universal reduction formalism.}
In order to present our contribution in total generality, we take the concept of $\Pred$ (introduced in \cite{BeyersdorffBC16} for a hierarchy of Frege systems) and formalise conditions on $\PS$ yielding a sound and complete QBF proof system (Theorem~\ref{thm:pred-proof-system}).
We identify three natural properties that are sufficient:
(a) The set of derivable axioms is semantically equivalent to the input formula;
(b) The system exhibits logical correctness and implicational completeness;
(c) The system is closed under restrictions.
Any line-based propositional calculus possessing all three properties is referred to as a \emph{base system} (Definition~\ref{def:base-system}).
Formalising the framework of base systems renders our technique applicable to the complete spectrum of $\Pred$ systems.
All the concrete propositional calculi considered in this work (i.e. those appearing in Figure~\ref{fig:proof-systems}) are demonstrably base systems.
We note that static propositional proof systems (such as Nullstellensatz) cannot be upgraded to $\Pred$ simply because they are not line-based.

\subsubsection{A new technique for genuine QBF lower bounds}
Our technique is based on careful construction and analysis of \emph{strategy extraction} in $\Pred$ \cite{BeyersdorffBC16}, a process by which a winning strategy for the universal player can be efficiently obtained from a refutation.
Given a $\Pred$ refutation $\pi$ of a QBF $\Phi$, \emph{round-based strategy extraction} works by first restricting $\pi$ according to the $\exists$-player's move, then collecting the response for the $\forall$-player from some line in $\pi$, and iterating until the evaluation game concludes. 


We define two measures called \emph{cost} and \emph{capacity} (Definitions~\ref{def:cost} and~\ref{def:capacity}). 
The cost of $\Phi$ is defined such that any winning strategy contains at least $\cost(\Phi)$ responses to some universal block.
Cost, therefore, is a natural semantically-grounded measure that provides a lower bound on the total number of extracted responses.
The upper bound is given by the capacity of $\pi$, a measure defined such that any response contributed from a given line in $\pi$ may be selected from a set of cardinality at most $\capa(\pi)$.  

Putting the two measures together, we obtain our main result, the \emph{Size-Cost-Capacity Theorem} (Theorem~\ref{thm:costcap}), which states that the size of $\pi$ is at least $\cost(\Phi) / \capa(\pi)$.
We also show explicitly that Size-Cost-Capacity returns a lower bound on the size of a \emph{semantic} $\Pred$ refutation, which illustrates that all results obtained by application of our technique are genuine QBF lower bounds in the aforementioned sense.

For direct applications of Size-Cost-Capacity, we first supply upper bounds on the capacity of refutations in concrete $\Pred$ proof systems.
We prove that all $\QURes$ and $\CPred$ refutations have capacity equal to~$1$ (Propositions~\ref{prop:QURes-capacity} and \ref{prop:CP-cap}), whereupon \emph{cost alone} gives an absolute proof-size lower bound (Corollaries~\ref{cor:QURes-bound} and~\ref{cor:CP-bound}).
The case for the QBF version of Polynomial Calculus with Resolution ($\PCRred$) is much more challenging, and requires some linear algebra, owing to the underlying algebraic composition of Polynomial Calculus.
Interestingly, it turns out that the capacity of a refutation in that system is no greater than its size (Proposition~\ref{prop:PC-capacity}), thus proof size is at least the square root of cost (Corollary~\ref{cor:PC-bound}).
Equipped with these corollaries, showing that the cost of a QBF is superpolynomial yields immediate proof-size lower bounds for all three systems simultaneously.

\subsubsection{Applications of the technique}
We demonstrate the effectiveness of our new technique on three applications.

\subsubsection*{A. The equality formulas: a non-trivial special case.}
We introduce a new family of hard QBFs that we call the \emph{equality formulas} (Definition~\ref{def:equality}), so called because the only winning strategy for the $\forall$-player in the evaluation game is to copy exactly the moves of the $\exists$-player.
We first prove that the equality formulas have exponential cost (Proposition~\ref{prop:equality-cost}).
Using Size-Cost-Capacity, we therefore prove that they require exponential-size refutations in $\QURes$, $\CPred$ and $\PCRred$.
We also demonstrate that the formulas have linear-size refutations in $\Fred$ (Proposition~\ref{prop:frege}), which shows that $\Fred$ proofs can have exponential capacity.

Closer inspection reveals that this lower bound is of a very specific type -- it is a genuine QBF lower bound (the formulas are not harbouring propositional hardness) that does not derive from a circuit lower bound (the winning strategy is not hard to compute in an associated circuit class).
In existing QBF literature, the only other example of such a family comes from the famous formulas of Kleine B\"{u}ning et al. \cite{BuningKF95} (cf.\ item C. below).
Those formulas are significantly more complex, and exhibit unbounded quantifier alternation compared to the (bounded) $\Sigma_3$ prefix of the equality formulas.

Indeed, the equality formulas appear to capture rather well the role that high cost plays in round-based strategy extraction to enforce large refutations.
For that reason, we include a direct proof of hardness for $\QURes$ in Section~\ref{sec:equality}.
This intended as a primer, to illustrate by example the concept and method-of-proof behind our lower bound technique.

\subsubsection*{B. The first hard random QBFs.}
For the major application of our technique, we define a class of random QBFs (Definition~\ref{def:Qnmc}) and prove that they are hard with high probability in all three systems $\QURes$, $\CPred$ and $\PCRred$ (Theorem~\ref{thm:random-hardness}).
We generate instances that combine the overall structure of the equality formulas with the existing model of random QBFs \cite{ChenI05} used in the competitive evaluation of solvers \cite{Pulina16,BrummayerLB10}. Thus, while at first glance these formulas may seem rather structured, they are simply a disjunction of QBFs generated using a minor variation on the (1,2)-QCNF model of \cite{ChenI05}.

Drawing on the existing literature \cite{Vega01j,ChvatalR92,CreignouDER15}, we show that suitable choices of parameters force our generated formulas to be false and have exponential cost with high probability (Lemma~\ref{lem:rand-highcost}).
Perhaps surprisingly, the cost lower bound is constructed by applying results on the unsatisfiability of random 2-SAT instances \cite{Vega01j} and the truth of random (1,2)-QCNFs \cite{CreignouDER15}.
These results both concern only the truth value of the corresponding formulas, and taken individually seem unrelated to cost.
However, by carefully choosing the number of clauses so as to allow the application of both results, we can construct a cost lower bound using a novel semantic argument.

Our contribution constitutes the first proof-size lower bounds for randomly generated formulas in the QBF proof complexity literature.
We emphasize that these are genuine QBF lower bounds in the aforementioned sense; they are not merely hard random CNFs lifted to QBF.  

\subsubsection*{C. New proofs of known lower bounds.}
Our final application uses Size-Cost-Capacity to provide a new proof of the hardness of the prominent QBFs of Kleine B\"{u}ning, Karpinski and Fl\"{o}gel \cite{BuningKF95}.
We consider a common modification of the formulas, denoted by $\lambda(n)$, in which each universal variable is `doubled'.
This modification is known to lift lower bounds in $\QRes$ to lower bounds in $\QURes$ \cite{BalabanovWJ14}, where we can apply Size-Cost-Capacity.

By rearranging the quantifier prefix to quantify all the additional universal variables in the penultimate quantifier block, we obtain a weaker formula with exponential cost (Proposition~\ref{prop:kb-cost}).
An elementary reduction then shows that $\lambda(n)$ requires exponential size $\QURes$ refutations (Corollary~\ref{cor:kb-size}).
As $\QURes$ lower bounds on these modified formulas are shown to be equivalent to $\QRes$ lower bounds on the original formulas, our technique even proves the original lower bounds from \cite{BuningKF95} (cf.\ also \cite{BeyersdorffCJ15}), and provides some insight as to the source of hardness.

\subsection{Relation to previous work}
It is fair to say that there is a scarcity of general methods for showing genuine lower bounds in systems like $\Pred$.
In contrast, a number of techniques for propositional calculi have emerged from the intense study of resolution \cite{Buss12,Segerlind07}.

Researchers have of course attempted to lift these techniques to quantified logic, but with mixed success.
The seminal size-width relations for resolution \cite{Ben-SassonW01}, which describe proof size in terms of proof width, are rendered ineffectual by universal quantification \cite{BeyersdorffCMS18}.  
The prover-delayer techniques of \cite{BeyersdorffGL13,PudlakI00} have been successfully lifted to QBF, but only apply to the weaker tree-like systems \cite{BeyersdorffCS17}, whereas solving techniques such as QCDCL are based on the stronger DAG-like versions.
Feasible interpolation \cite{Krajicek97} is an established propositional technique that has been successfully adapted \cite{BeyersdorffCMS17}, but it is applicable only to a small class of hand-crafted QBFs of a rather specific syntactic form.

Strategy extraction for QBF lower bounds has been explored previously by exploiting connections to circuit complexity \cite{BeyersdorffCJ15,BeyersdorffBC16,BeyersdorffP16}.
In particular, \cite{BeyersdorffBC16} established tight relations between circuit and proof complexity, lifting even strong circuit lower bounds for $\AC{0}[p]$ circuits \cite{Razborov87,Smolensky87} to QBF lower bounds for $\AC{0}[p]{-}\Fred$ \cite{BeyersdorffBC16}, which is unparalleled in the propositional domain.
In fact, for strong proof systems such as $\Fred$, this strategy extraction technique is sufficient to prove any genuine QBF lower bound, in the sense that any superpolynomial lower bound for $\Fred$ arises either due to a lower bound for $\Frege$, or due to a lower bound for Boolean circuits \cite{BeyersdorffP16}.
However for weaker systems such as $\QURes$, this does not hold and there exist lower bounds which are neither a propositional lower bound nor a circuit lower bound \cite{BeyersdorffHP17}.
The reasons underlying such hardness results are at present not well understood.
The development of techniques for, or a characterisation of, such lower bounds would be an important step in QBF proof complexity.

The major drawback of the existing approach of \cite{BeyersdorffCJ15,BeyersdorffBC16,BeyersdorffP16}, of course, is the rarity of superpolynomial lower bounds from circuit complexity \cite{Vollmer99}, especially for larger circuit classes to which the stronger QBF proof systems connect.
With Size-Cost-Capacity we employ a much different approach to strategy extraction.
Our technique is motivated by semantics and \emph{does not interface with circuit complexity whatsoever}.
Instead, lower bounds are determined directly from the semantic properties of the instance, and consequently we make advances out of the reach of previous techniques.  

\subsection{Organisation of the paper} %
We provide the relevant background in Section~\ref{sec:prelims}.
In Section~\ref{sec:equality}, we introduce the equality formulas and give the direct proof of hardness for $\QURes$, while providing an overview of round-based strategy extraction.
Section~\ref{sec:framework} introduces our framework, including the formal definition of a base system and the proofs of soundness and completeness of $\Pred$.
Our lower bound technique follows in Section~\ref{sec:technique}, comprising definitions of cost and capacity and the proof of the Size-Cost-Capacity Theorem.
Upper bounds on capacity for $\CPred$ and $\PCRred$ are the subject of Section~\ref{sec:cap-bounds}.
Our major application to random QBFs is presented in Section~\ref{sec:random-QBFs}.
In Section~\ref{sec:kleinebuening} we apply Size-Cost-Capacity to the formulas of Kleine B\"{u}ning et al. \cite{BuningKF95}.
We offer some concluding thoughts in Section~\ref{sec:conclusions}. 

\section{Preliminaries} %
\label{sec:prelims}
\subsection{Quantified Boolean formulas}
A \emph{conjunctive normal form} (CNF) formula is a conjunction of clauses, each of which is a disjunction of literals, and a \emph{literal} is a Boolean variable or its negation.
We represent a CNF as a set of clauses, and a clause as a set of literals.

A \emph{quantified Boolean formula} (QBF) in \emph{closed prenex form} is typically denoted $\Phi :=\QBF{\Q}{\phi}$.
In the \emph{quantifier prefix} $\Q := \Q_1 X_1 \cdots \Q_n X_n$, the $X_i$ are pairwise-disjoint sets of Boolean variables (or \emph{blocks})\footnote{Whereas a block $X = \{x_1, \dots ,x_m\}$ is a set, it is written explicitly in a prefix as a string of variables $x_1 \cdots x_m$.} each of which is quantified either existentially or universally by the \emph{associated quantifier} $\Q_i \in \{\exists , \forall\}$, and consecutive blocks are oppositely quantified.
The \emph{propositional part} $\phi$ is a propositional formula all of whose variables $\vars(\phi)$ are quantified in \Q.
A \emph{QCNF} is a QBF whose propositional part is a CNF.

For a literal $l$,  we write $\var(l) := x$ iff $l = x$ or $l = \neg x$.
By the variables of $\Phi$ we mean the set $\vars(\Phi) := \bigcup^n_{i=1} X_i$. 
The set of existential variables of $\Phi$, denoted $\vars_\exists(\Phi)$, is the union of those $X_i$ whose associated quantifier $\Q_i$ is $\exists$, and we define the universal variables of $\Phi$ similarly. 
The prefix \Q defines a binary relation $<_\Q$ on the variables of $\Phi$, such that $x_i <_\Q x_j$ holds iff $x_i \in X_i$, $x_j \in X_j$ and $i<j$, in which case we say that $x_i$ \emph{is left of} $x_j$ ($x_j$ \emph{is right of} $x_i$) with respect to \Q.
For two variable sets $X,X^\prime \subseteq \vars(\Phi)$, we write $X <_\Q X^\prime$ iff each variable in $X$ is left of each variable in $X^\prime$ with respect to \Q.

A \emph{literal} $l$ is a Boolean variable $x$ or its negation $\lnot x$, and we write $\var(l) := x$.
A \emph{total assignment} $\tau$ to a set $\vars(\tau) = X$ of Boolean variables is a function $\tau:X \rightarrow \{0,1\}$, typically represented as a set of literals in which the literal $\lnot x$ (resp. $x$) represents the assignment $x \mapsto 0$ (resp. $x \mapsto 1$).
The set of all total assignments to $X$ is denoted $\langle X \rangle$.
A \emph{partial assignment} to $X$ is a total assignment to a subset of $X$.
The \emph{projection} of $\tau$ to a set $X^\prime$ of Boolean variables is the assignment $\{l \in \tau : \var(l) \in X^\prime\}$.

The \emph{restriction} of $\Phi$ by an assignment $\tau$ is $\Phi[\tau] := \QBF{\Q[\tau]}{\phi[\tau]}$, where $\Q[\tau]$ is obtained from $\Q$ by removing each variable in $\vars(\tau)$ (and any redundant quantifiers), and $\phi[\tau]$ is the restriction of $\phi$ by $\tau$.
Restriction of propositional formulas is defined by the conventional inductive semantics of propositional logic; that is, $\phi[\tau]$ is obtained from $\phi$ by substituting each occurrence of a variable in $\vars(\tau)$ by its associated truth value, and simplifying the resulting formula in the usual way. 
%
\subsection{QBF semantics}
Semantics are neatly described in terms of strategies in the two-player \emph{evaluation game}.
The game takes place over $n$ rounds, during which the variables of a QBF $\Phi := \QBF{\Q}{\phi}$ are assigned strictly in the linear order of the prefix $\Q := \exists E_1 \forall U_1 \cdots \exists E_n \forall U_n$.\footnote{An arbitrary QBF can be written in this form by allowing $E_1$ and $U_n$ to be empty.}
In the $i^{\mbox{\scriptsize th}}$ round, the existential player selects an assignment $\alpha_i$ to $E_i$ and the universal player responds with an assignment $\beta_i$ to $U_i$.
At the conclusion the players have constructed a total assignment $\tau := \bigcup^n_{i=1} (\alpha_i \cup \beta_i) \in \langle \vars(\Phi) \rangle$.
The existential player wins iff $\phi[\tau] = \top$; the universal player wins iff $\phi[\tau] = \bot$. 

A strategy for the universal player details exactly how she should respond to all possible moves of the existential player.
Formally, a \emph{$\forall$-strategy} for $\Phi$ is a function $S : \langle \vars_\exists(\Phi) \rangle \rightarrow \langle \vars_\forall(\Phi) \rangle$ that satisfies the following for each $\alpha, \alpha^\prime \in \dom(S)$ and each $i \in [n]$: if $\alpha$ and $\alpha^\prime$ agree on $E_1 \cup \cdots \cup E_i$, then $S(\alpha)$ and $S(\alpha^\prime)$ agree on $U_1 \cup \cdots \cup U_i$.\footnote{Two assignments agree on a set if and only if their projections to that set are identical.}
We say that $S$ is \emph{winning} iff $\phi[\alpha \cup S(\alpha)] = \bot$ for each $\alpha \in \vars_\exists(\Phi)$.

Existential strategies are defined dually.
Given a QBF $\Phi^\prime := \QBF{\forall U_1 \exists E_1 \cdots \forall U_n \exists E_n}{\phi}$, an \emph{$\exists$-strategy} for $\Phi^\prime$ is a function $S^\prime : \langle \vars_\forall(\Phi^\prime) \rangle \rightarrow \langle \vars_\exists(\Phi^\prime) \rangle$ that satisfies the following for each $\beta, \beta^\prime \in \dom(S^\prime)$ and each $i \in [n]$:
if $\beta$ and $\beta^\prime$ agree on $U_1 \cup \cdots \cup U_i$, then $S(\beta)$ and $S(\beta^\prime)$ agree on $E_1 \cup \cdots \cup E_i$.
We say that $S$ is \emph{winning} iff $\phi[\beta \cup S^\prime(\beta)] = \top$ for each $\alpha \in \vars_\forall(\Phi^\prime)$.

\begin{prop}[folklore]
A QBF is false if and only if it has a winning $\forall$-strategy, and is true if and only if it has winning $\exists$-strategy.
\end{prop}

%
\subsection{QBF resolution}
\emph{Resolution} is a well-studied refutational proof system for propositional CNF formulas with a single inference rule: the \emph{resolvent} $C_1 \cup C_2$ may be derived from clauses $C_1 \cup \{x\}$ and $ C_2 \cup \{\lnot x\}$.
Resolution is \emph{refutationally} sound and complete: that is, the empty clause can be derived from a CNF iff it is unsatisfiable.
Resolution becomes implicationally complete with the addition of the weakening rule, which allows literals to be added to clauses arbitrarily. 

\emph{QU-Resolution} ($\QURes$) \cite{BuningKF95,Gelder12} is a resolution-based proof system for false QCNFs.
The calculus supplements resolution with a \emph{universal reduction rule} which allows (literals in) universal variables to be removed from a clause $C$ provided that they are right of all existentials in $C$ with respect to $\Q$.
Tautological clauses are explicitly forbidden; for any variable $x$, one may not derive a clause containing both $x$ and $\lnot x$.
The rules of $\QURes$ are given in Figure~\ref{fig:QURes}.
Note that we choose to include weakening of clauses as a valid inference rule, to emphasize the implicational completeness of the underlying propositional system.

A \emph{$\QURes$ derivation} of a clause $C$ from a QCNF $\Phi$ is a sequence $C_1, \dots, C_m$ of clauses in which (a) each $C_i$ is either introduced as an axiom (i.e. $C_i \in \phi$) or is derived from previous clauses in the sequence using resolution or universal reduction, and (b) the \emph{conclusion} $C = C_m$ is the unique clause that is not an antecedent in the application of one of these inference rules.
A \emph{refutation} of $\Phi$ is a derivation of the empty clause from $\Phi$.

\begin{figure}[t]
\QUResDerivationRules
\caption{The rules of QU-Resolution. The input is a QCNF $\Phi = \QBF{\Q}{\phi}$ whose propositional part contains no tautological clauses.\label{fig:QURes}}
\end{figure}

\section{A direct proof of hardness for the equality formulas} %
\label{sec:equality}
In this section, we introduce the \emph{equality formulas} and give a direct proof of their hardness in the well-known QBF proof system $\QURes$.
The material in this section is intended to illuminate, by means of an accessible example, the paradigm of round-based strategy extraction, and our exploitation of it as a new lower-bound technique.

\subsection{The equality formulas.}
\label{subsec:equality}
The salient feature of the equality formulas, defined below, is that each instance has a unique winning strategy, and the cardinality of its range is exactly $2^n$.

\begin{defi}[equality formulas]
\label{def:equality}
For $n \in \mathbb{N}$, the $n^{\mbox{\scriptsize th}}$ \emph{equality formula} is
\begin{equation*}
\eq(n) := \exists x_1 \cdots x_n \forall u_1 \cdots u_n \exists t_1 \cdots t_n \cdot \left( \bigwedge^n_{i=1} (x_i \vee u_i \vee \lnot t_i) \wedge (\lnot x_i \vee \lnot u_i \vee \lnot t_i)\right) \wedge \left( \bigvee^n_{i=1} t_i \right).
\end{equation*}
\end{defi}

\noindent Note that the propositional part of $\eq (n)$ is the CNF consisting of the \emph{long clause} $\{t_1, \dots , t_n\}$ and each pair of clauses $\{x_i, u_i,\lnot  t_i\},\{\lnot x_i,\lnot u_i, \lnot t_i\}$ for $i \in [n]$.

The equality formulas are false, and it is clear that there is only one winning strategy for the universal player; namely, she must assign each $u_i$ the same value as the corresponding $x_i$.
Proceeding this way, she forces all $n$ unit clauses $\{\lnot t_i\}$ to be present on the board with only the final block left to play.
Then the existential player must lose, since satisfying all such unit clauses entails falsifying the long clause $\{t_1, \dots ,t_n\}$.
This is indeed the only way to win, since any other reply from the universal player would drop at least one unit clause, allowing her opponent to satisfy the long clause.

The upshot is that the existential player can \emph{force} his opponent to play any one of the total assignments to the universal variables.
It follows that the range of the unique winning $\forall$-strategy for $\eq(n)$ is exactly the set $R := \langle \{u_1, \dots ,u_n\} \rangle$.
The fact that this set has exponential cardinality is a key feature that we exploit in our lower bound proof.


\subsection{Overview of round-based strategy extraction.}
\label{subsec:strategy-extraction}
Strategy extraction is an important QBF paradigm that was motivated by solving certification (cf. \cite{NarizzanoPPT09,GoultiaevaGB11}), and subsequently received much attention in the literature \cite{BalabanovJ12,PeitlSS16,EglyLW13,BeyersdorffBC16}.
In this paper, we follow the algorithm given in \cite{GoultiaevaGB11}, which for the sake of clarity we refer to as \emph{round-based strategy extraction}.

Given a $\QURes$ refutation of a QCNF $\QBF{\exists E_1 \forall U_1 \dots \exists E_n \forall U_n}{\phi}$, round-based strategy extraction is an iterative procedure that computes a winning $\forall$-strategy.
During the course of the game, the $\forall$-player maintains a restriction of the refutation, from which her winning responses may be determined.
In each round she performs two operations:
\begin{itemize}[labelwidth=10cm]
\item restrict the current refutation by the $\exists$-player's move $\alpha_i \in \langle E_i \rangle$;
\item `read off' a response $\beta_i \in \langle U_i \rangle$ from the current refutation, then restrict by $\beta$.

\end{itemize}
These operations are repeated round by round until the game concludes.
At termination, we obtain a refutation of the input formula under a total assignment.
Since $\QURes$ is sound, the assignment must falsify the propositional part of the formula.
Hence, the correctness of the procedure rests on the fact that restrictions by $\alpha_i$ and $\beta_i$ preserve the refutation.
This is detailed below over Propositions~\ref{prop:QURes-existential-restriction} and~\ref{prop:QURes-universal-restriction}.

To streamline the material, we work with a normal form of $\QURes$ derivation in which universal reduction has stronger side conditions: one must remove all the universal literals from the rightmost universal block appearing in the clause.
Formally, one may derive $C$ from $C \cup R$ by universal reduction provided that $\vars(C) <_\Q \vars(R)$ and $\vars(R)$ is a subset of some universal block $U$ in $\Q$, where $\Q$ is the prefix of the input QBF.

It is easy to see that $\QURes$ derivations can be placed in such a normal form with no increase in size.
To do so, one first applies as large a universal reduction as possible to every clause, yielding a new refutation in which each clause is a subset of the corresponding one from the original refutation.
One then takes the first occurrence of the empty clause as the conclusion, and keeps only the subderivation of this conclusion, discarding all other clauses (as well as any duplicates), which are rendered redundant.

For the rest of this section, we assume $\QURes$ derivations to be in this normal form.
It is easy to see that normal refutations need \emph{at most one universal} reduction on the leftmost block, since the consequent of such a reduction must be the empty clause.

\subsubsection*{Restriction by the $\exists$-player's move.}
Informally, the restriction of a $\QURes$ refutation $\pi$ by an arbitrary assignment $\delta$ (denoted $\pi[\delta]$) is obtained simply by restricting the axiom clauses by $\delta$, discarding any that are satisfied, and carrying out the steps of the orignal refutation wherever possible.
Under a formal definition of restriction, it is easy to prove that existential restrictions preserve $\QURes$ refutations. 

\begin{prop}
\label{prop:QURes-existential-restriction}
Let $\pi$ be a $\QURes$ refutation of a QCNF $\Phi$, and let $\alpha$ be a partial assignment to the existential variables of $\Phi$.
Then $\pi[\alpha]$ is a $\QURes$ refutation of $\Phi[\alpha]$.
\end{prop}

It follows immediately that restriction by the $\exists$-player's move $\alpha_i \in \langle E_i \rangle$ preserves the current refutation.

\subsubsection*{Restriction by the $\forall$-player's response.}
Arbitrary universal restrictions do not preserve $\QURes$ refutations, as it is possible that a universal assignment satisfies the antecedent of a universal reduction step, but not the consequent.
However, a universal restriction does indeed preserve a refutation provided it does not satisfy any reduced literals.
If the first block $U$ of $\Q$ is universal, then we have at most one universal reduction on $U$, since the refutation finishes after this step.
Therefore, restriction by a total assignment to $U$ need only be consistent with the literals reduced in this step (if it exists) in order to preserve the refutation.
For non-trivial refutations, all of these reduced literals appear in the penultimate clause.

\begin{prop}
\label{prop:QURes-universal-restriction}
Let $\pi$ be a non-trivial $\QURes$ refutation of a QCNF $\Phi$ whose first block $U$ is universal, let $\beta \in \langle U \rangle$, and let $C$ be the set of $U$-literals appearing in the penultimate clause of $\pi$. 
If $\{\lnot l : l \in C\} \subseteq \beta$, then $\pi[\beta]$ is a refutation of $\Phi[\beta]$.
\end{prop}

Hence, if the current refutation is non-trivial, the $\forall$-player may read off the $U$-literals from the penultimate clause, negate them, and extend them arbitrarily to an assignment $\beta_i \in \langle U_i \rangle$.
Preservation of the current refutation under restriction by $\beta_i$ is guaranteed by Proposition~\ref{prop:QURes-universal-restriction}.
If the current refutation is trivial, an arbitrary $\beta_i$ is clearly sufficient.

\subsection{Direct proof of hardness.}
\label{subsec:direct-proof}
Propositions~\ref{prop:QURes-existential-restriction} and~\ref{prop:QURes-universal-restriction} form the basis of a lower-bound proof for the equality formulas, via the following lemma.

\begin{lem}
\label{lem:QURes-assignments}
Let $n \in \mathbb{N}$, let $\pi$ be a $\QURes$ refutation of $\eq(n)$, and let $\beta$ be a total assignment to the universal variables of $\eq(n)$.
Then there exists a clause in $\pi$ that contains $\{\lnot l : l \in \beta \}$ as a subset.
\end{lem}

\begin{proof}
Let $X := \{x_1, \dots, x_n\}$, let $U := \{u_1, \dots , u_n\}$, and let $\alpha$ be the total assignment to $X$ defined by
$$\alpha(x_i) := \beta(u_i), \mbox{\,for each } i \in [n].$$
Observe that, by the nature of the equality formulas, $\beta$ is the unique total assignment to $U$ under which $\eq(n)[\alpha]$ is false. 
Moreover, $\pi[\alpha]$ is non-trivial, since the propositional part of $\eq(n)[\alpha]$ does not contain the empty clause.
Let $C$ be the set of $U$-literals appearing in the penultimate clause of $\pi[\alpha]$, and let $\beta^\prime \in \langle U \rangle$.
Now, if $\{\lnot l : l \in C\} \subseteq \beta^\prime$, then $\eq(n)[\alpha \cup \beta^\prime]$ is false, by Proposition~\ref{prop:QURes-universal-restriction} and the soundness of $\QURes$.
Therefore, if $\beta^\prime$ extends $\{\lnot l : l \in C\}$, then $\beta^\prime = \beta$.
It follows that $\beta = \{\lnot l : l \in C\}$, or equivalently, $C = \{\lnot l : l \in \beta\}$
The lemma follows since any clause in $\pi[\alpha]$ is a subset of some clause in $\pi$, by definition of restriction of clauses.
\end{proof}

Since $\QURes$ disallows tautological clauses, the lower bound is an easy consequence of the preceding lemma.

\begin{thm}
\label{thm:QURes-equality}
Let $n \in \mathbb{N}$, and let $\pi$ be a $\QURes$ refutation of $\eq(n)$.
Then $|\pi| \geq 2^n$.
\end{thm}

\begin{proof}
Aiming for contradiction, suppose that $|\pi| < 2^n$.
Observe that there are exactly $2^n$ total assignments to $\{u_1, \dots , u_n\}$.
Then, by Lemma~\ref{lem:QURes-assignments}, there exist distinct total assignments $\beta_1,\beta_2 \in \langle\{u_1, \dots , u_n\} \rangle$ and a clause $C$ in $\pi$ such that $\{\lnot l : l \in \beta_1 \}$ and $\{\lnot l : l \in \beta_2 \}$ are both subsets of $C$.
Since $\beta_1$ and $\beta_2$ are distinct, there is some literal $l$ satisfying $l \in \beta_1$ and $\lnot l \in \beta_2$.
But then $C$ is a tautological clause, and does not appear in any $\QURes$ derivation.
\end{proof}

In a nutshell, Theorem~\ref{thm:QURes-equality} was proved by equating the minimum refutation size with the cardinality of the range of a winning $\forall$-strategy for $\eq(n)$.
Our argument here was aided by two facts: $\eq(n)$ has a unique winning $\forall$-strategy and contains a single universal block.
Of course, neither fact holds for QBFs in general.
Nonetheless, in the sequel we generalise the method to prove an absolute proof-size lower bound for any instance in $\Pred$.

\section{Our framework} %
\label{sec:framework}

In this section, we develop a framework for Size-Cost-Capacity centred on a precise description of the $\Pred$ formalism.
In Subsection~\ref{subsec:line-based}, we first describe what we mean by a `line-based' propositional proof system $\PS$, and then introduce the notion of a `base system' - a line-based system satisfying some natural conditions.
In Subsection~\ref{subsec:base-systems} we define $\Pred$ with respect to a line-based system $\PS$, and prove that it is sound and complete if $\PS$ is a base system.

\subsection{Propositional base systems}
\label{subsec:line-based}
We associate the basic concept of a \emph{line-based propositional proof system} $\PS$ with the following two features:
\begin{itemize}
\item A set of \emph{lines} $\mathcal{L}_\PS$, containing at least the two lines $\top$ and $\bot$ that represent trivial truth and trivial falsity, respectively.
\item A set of \emph{inference rules} $\mathcal{I}_\PS$ and an \emph{axiom function} that maps each propositional formula $\phi$ to a set of axioms $\mathcal{A}_\PS(\phi) \subseteq \mathcal{L}_\PS$.
The axiom function should be polynomial-time computable, and the validity of applications of inference rules should be polynomial-time checkable.
\end{itemize}

Following convention, a $\PS$-derivation from a propositional formula $\phi$ is a sequence $\pi = L_1, \dots ,L_m$ of lines from $\mathcal{L}_\PS$, in which each line $L_i$ is either an axiom from the set $\mathcal{A}_\PS(\phi)$, or may be derived from previous lines using an inference rule in $\mathcal{I}_\PS$.
The final line $L_m$ is called the \emph{conclusion} of $\pi$, and $\pi$ is a refutation iff $L_m = \bot$.

In order to facilitate the restriction of $\PS$-derivations, we require two further features: 
\begin{itemize}
\item A \emph{variables function} that maps each line $L \in \mathcal{L}_\PS$ to a finite set of Boolean variables $\vars(L)$, satisfying $\vars(\top) = \vars(\bot) = \emptyset$. Additionally, $\vars(L) \subseteq \vars(\phi)$ for each line $L$ in a $\PS$-derivation from $\phi$.\footnote{Note that this does not exclude extended Frege systems ($\mathsf{EF}$), whose lines can be represented as Boolean circuits as in \cite[p. 71]{Jerabek:phd-thesis}.}
\item A \emph{restriction operator} (denoted by square brackets) that takes each line $L \in \mathcal{L}_\PS$, under restriction by any partial assignment $\tau$ to $\vars(L)$, to a line $L[\tau] \in \mathcal{L}_\PS$.
If $\tau$ is a total assignment, then $L[\tau]$ is either $\top$ or $\bot$.
Restriction of $L$ by an arbitrary Boolean assignment $\sigma$ is defined as the restriction of $L$ by the projection of $\sigma$ to $\vars(L)$.
\end{itemize}

It should be clear that the purpose of the restriction operator is to encompass the natural semantics of $\PS$.
For that reason, we made the natural stipulation that restriction by a total assignment to the variables of a line yields either trivial truth or trivial falsity.
We may therefore associate with any line $L \in \mathcal{L}_\PS$ the Boolean function on $\vars(L)$ that computes the propositional models of $L$, with respect to the semantics of the restriction operator for $\PS$.
%
\begin{defi}[associated Boolean function]
Let $\PS$ be a line-based propositional proof system and let $L \in \mathcal{L}_\mathsf{P}$.
The \emph{associated Boolean function} for $L$ is $B_L : \langle \vars(L) \rangle \rightarrow \{0,1\}$, defined by 
\begin{equation*}
B_L(\tau) =
\begin{cases}
1, & \mbox{if }L[\tau] = \top\,,\\
0, & \mbox{if }L[\tau] = \bot\,.
\end{cases}
\end{equation*}
\end{defi}

It is useful to define the usual notion of semantic entailment on the lines of $\PS$.
Given two lines $L, L^\prime \in \mathcal{L}_\PS$, we say that $L^\prime$ \emph{semantically entails} $L$ if 
$$L^\prime[\tau] = \top \quad \Rightarrow \quad L[\tau] = \top\,, \quad \mbox{for each } \tau \in \langle \vars(L) \cup \vars(L^\prime) \rangle\,,$$
and $L$ and $L^\prime$ are said to be \emph{semantically equivalent} (written $L \equiv L^\prime$) if they semantically entail one another. 
We say that a finite set $\{L_1, \dots ,L_k\} \subseteq \mathcal{L}_\PS$ of lines semantically entails $L$ if
$$L_i[\tau] = \top \mbox{ for each } i \in [k] \quad \Rightarrow \quad L[\tau] = \top\,, \quad\mbox{for each } \tau \in \langle \bigcup_{i \in [n]}\vars(L_i) \cup \vars(L) \rangle \,.$$

Beyond the notion of \emph{line-based}, we identify three natural properties.
The first of these guarantees that the propositional models of the axioms are exactly those of the input formula, and the second guarantees soundness and completeness \emph{in the classical sense of propositional logic}.\footnote{The (proof-complexity-theoretic) concepts of soundness and completeness for arbitrary proof systems in the sense of Cook and Reckhow are weaker than their counterparts in propositional logic.}
The third property ensures that the restriction operator behaves sensibly; that is, the propositional models of the restricted line are computed by the restriction of the associated Boolean function.
We introduce the term \emph{base system} for those possessing all three. 
\begin{defi}[base system]
\label{def:base-system}
A \emph{base system} $\PS$ is a line-based propositional proof system satisfying the following three properties:
\begin{itemize}
\item
\emph{Axiomatic equivalence}. For each propositional formula $\phi$ and each $\tau \in \langle \vars(\phi) \rangle$, $\phi[\tau] = \top$ iff each $A \in \mathcal{A}_{\mathsf{P}}(\phi)$ satisfies $A[\tau] = \top$.
\item
\emph{Inferential equivalence}. For each finite set of lines $\mathcal{L} \subseteq \mathcal{L}_\mathsf{P}$ and each line $L \in \mathcal{L}_\mathsf{P}$,
$L$ can be derived from $\mathcal{L}$ iff $\mathcal{L}$ semantically entails $L$.
\item
\emph{Restrictive closure}. For each $L \in \mathcal{L}_\mathsf{P}$ and each partial assignment $\tau$ to $\vars(L)$, the Boolean functions $B_{L[\tau]}$ and $B_L|_\tau$ are identical.
\end{itemize}
\end{defi}

As a first example, we note that resolution (with weakening, and excluding tautological clauses) forms a base system. 
The axiomatic equivalence is trivial, as is inferential equivalence, which follows directly from implicational completeness\footnote{Formally, one should introduce a special clause $\top$ that can be derived from the empty set of clauses.} and logical correctness.
Taking the conventional definitions of the variable function and restriction operator, the restrictive closure of resolution is readily verified.
This is to be expected of course, since the restriction of clauses is based on a standard definition of semantics in propositional logic.

\subsection{The $\Pred$ formalism}
\label{subsec:base-systems}

Universal reduction is a widely used rule of inference in QBF proof systems, by which universal variables may be assigned under certain conditions.
More precisely, a line $L$ may be restricted by a partial assignment to a universal block $U$ provided it is \emph{rightmost} in $\vars(L)$, with respect to the prefix $\Q$ of the input QBF; that is, $U$ is right of every block in $\Q$ whose intersection with $\vars(L)$ is non-empty.
We state the rule formally in Figure~\ref{fig:universal-reduction}.
By restrictive closure, the restriction of a line by an assignment to $\beta$ results in the exclusion of $\vars(\beta)$ from the domain of the associated Boolean function.
Universal reduction should therefore be viewed as a sound method for deleting universal variables.

\begin{figure}[t]
\JLBFigUniversalReduction
\caption{The universal reduction rule, where $\Phi = \QBF{\Q}{\phi}$ is the input QBF.\label{fig:universal-reduction}}
\end{figure}

The purpose of universal reduction is to lift a propositional proof system $\PS$ to a QBF system $\Pred$, as in the following definition.

\begin{defi}[$\Pred$ \cite{BeyersdorffBC16}]
Let $\PS$ be a line-based propositional proof system.
Then $\Pred$ is the system consisting of the inference rules of $\PS$ in addition to universal reduction, in which references to the input formula $\phi$ in the rules of $\PS$ are interpreted as references to the propositional part of the input QBF $\QBF{\Q}{\phi}$.
\end{defi}

We lift some notation from $\PS$ to $\Pred$ as follows.
We denote the lines available in $\Pred$ (syntactically equivalent to the lines available in $\PS$) by $\mathcal{L}^\Q_\PS$, where $\Q$ is the prefix of the input QBF.
For $L \in \mathcal{L}^\Q_\PS$, we write $\vars_\exists(L)$ and  $\vars_\forall(L)$ for the subsets of $\vars(L)$ consisting of the existentially and universally quantified variables, respectively.


Soundness and completeness of $\Pred$ is not guaranteed for an arbitrary line-based system $\PS$, but it is guaranteed if $\PS$ is a base system.
The following lemma establishes completeness.

\begin{lem}
\label{lem:Pred-completeness}
Let $\PS$ be a base system.
Every false QBF has a $\Pred$ refutation.
\end{lem}

\begin{proof}
Let $\Phi = \QBF{\Q}{\phi}$ be a false QBF with universal variables $\vars_\forall(\Phi) = \{u_1, \dots ,u_n\}$, and let $\St$ be a winning $\forall$-strategy for $\Phi$.
For each $i \in [n]$, let $X_i$ be the set of existential variables left of $u_i$ in $\Q$.

For each $i \in [n]$, let $B_i$ be the Boolean function with domain $\langle X_i \cup \{u_1, \dots, u_i\} \rangle$ that maps to $0$ iff, for each $j \in [i]$, the assignment of $u_j$ matches that of the strategy on the assignment to $X_j$.
Formally, $B_i(\sigma) = 0$ iff $\sigma_\forall \subseteq \St(\sigma^\prime_\exists)$ for every extension of $\sigma_\exists$ to $\vars_\exists(\Phi)$, where $\sigma_\exists$ and $\sigma_\forall$ are the existential and universal subassignments of $\sigma$.
Intuitively, $B_i$ maps to $0$ iff the universal player plays according to the strategy $\St$ up to the variable $u_i$.
Also, let $B_0$ be the trivial Boolean function that maps the empty assignment to $0$.

It is an immediate consequence of the axiomatic equivalence of $\PS$ that, for any Boolean function $B$, there exists a set of lines in $\mathcal{L}_\PS$ whose conjunction has $B$ as its associated Boolean function. 
($\A_\PS(\phi_B)$, where $\phi_B$ is a propositional formula representing $B$, is one such set.)
In particular, for each $i = 0,1, \dots, n$, there exists a set of lines $\mathcal{L}_i$ whose conjunction has $B_i$ as its associated Boolean function.
Observe that $\mathcal{L}_i$ essentially encodes the statement that the universal player does \emph{not} play according to the strategy $\St$ up to the variable $u_i$.

By backwards induction on $i = 0, 1, \dots, n$, we show that $\Pred$ can derive each set $\mathcal{L}_i$.
At the final step $i = 0$, we hence prove the lemma, since $B_0$ is the identically zero Boolean function on an empty set of variables, which implies that $\mathcal{L}_0$ is semantically equivalent to $\bot$.

For the base case, observe that $\phi$ semantically entails $\mathcal{L}_n$, since $\mathcal{L}_n$ encodes the statement that the universal player does not play according the whole strategy $\St$.
Therefore in a $\Pred$ derivation from $\Phi$ one may derive each line in $\mathcal{L}_n$, by the axiomatic equivalence and inferential equivalence of $\PS$.

For the inductive step, let $i \in [n]$.
Since $u_i$ is the rightmost variable appearing in $L_i$, by universal reduction one may derive both sets of lines
$$\mathcal{L}^0_i := \{L[u_i \mapsto 0] : L \in \mathcal{L}_i\} \mbox{ \,and\, } \mathcal{L}^1_i := \{L[u_i \mapsto 1] : L \in \mathcal{L}_i\}\,.$$
It is readily verified that the union of $\mathcal{L}^0_i$ and $\mathcal{L}^0_i$ semantically entails $\mathcal{L}_{i-1}$, which may then be derived by the inferential equivalence of $\PS$.
This completes the inductive step, and the proof.
\end{proof}

Soundness of the $\Pred$ can in fact be proved for a relaxed definition of refutation in which any logically correct propositional inference is allowed, and lines introduced by universal reduction need only be semantically equivalent to the consequent from the original definition (Figure~\ref{fig:universal-reduction}).
A sequence satisfying these relaxed conditions we call a \emph{semantic} $\Pred$ refutation.

\begin{defi}[semantic $\Pred$ refutation]
\label{def:semantic-ref}
Let $\PS$ be a base system.
A \emph{semantic $\Pred$ refutation} of a QBF $\Phi = \QBF{\Q}{\phi}$ is a sequence $\pi = L_1, \dots ,L_m$ of lines from $\mathcal{L}_\PS$, in which $L_m = \bot$ and each $L_i$ satisfies at least one of the following:
\begin{itemize}
\item
\emph{Axiom}. $L_i \in \mathcal{A}_\PS(\phi)$
\item
\emph{Semantic consequence}. $\{L_1, \dots ,L_{i-1}\}$ semantically entails $L_i$.
\item
\emph{Semantic universal reduction}. $L_i \equiv L_j[\beta]$, where $j < i$, and $\beta$ is a partial assignment to a universal block $U$ of $\Q$ that is rightmost in $\vars(L_j)$.
\end{itemize}
\end{defi}

Semantic refutations feature in the following section, where they are used to show that the Size-Cost-Capacity Theorem returns a genuine refutation-size lower bound.
We introduce them at this point because the soundness of semantic refutations, proved in the following lemma, is used in the proof of that theorem.

\begin{lem}
\label{lem:Pred-soundness}
Let $\PS$ be a base system.
If a QBF has a semantic $\Pred$ refutation, then it is false.
\end{lem}

\begin{proof}
Let $L_1, \dots, L_m$ be a semantic $\Pred$ refutation of a QBF $\Phi = \QBF{\Q}{\phi}$.
Aiming for contradiction suppose that $\Phi$ is true, and let $S$ be a winning $\exists$-strategy for $\Phi$.
For each line $L_i$, we say that $S$ \emph{models} $L_i$ if $L_i[\beta \cup S(\beta)] = \top$ for each $\beta \in \langle \vars_\forall(\Phi) \rangle$.
By induction on $i \in [m]$, we show that $S$ models $L_i$.
At the final step $i = n$, we therefore reach a contradiction since $L_m = \bot$ cannot be modelled by any $\exists$-strategy.

Since $L_1$ is in $\A_\PS(\phi)$, the base case $i = 1$ follows immediately from the axiomatic equivalence of $\PS$.
For the inductive step, let $i \geq 2$ and suppose that $S$ models each $L_1, \dots, L_{i-1}$.
The case where $L_i$ is an axiom is identical to the base case.
If $L_i$ was derived by semantic consequence, then $L_1 \wedge \cdots \wedge L_{i-1}$ semantically entails $L_i$, and it is easy to see that $S$ models $L_i$.

It remains to consider the case where $L_i$ was derived by semantic universal reduction.
Then $L_i \equiv L_j[\beta_U]$, where $j < i$, and $\beta_U$ is a partial assignment to a universal block $U$ of $\Q$ that is rightmost in $\vars(L_j)$.
Let $\beta \in \langle \vars_\forall(\Phi)\rangle$, and let $\beta^\prime$ be the assignment obtained from $\beta$ by overwriting the assignments to $\vars(\beta_U)$ with those of $\beta_U$; that is,
$$ \beta^\prime := (\beta \setminus \{l \in \beta : \var(l) \in \vars(\beta_U)\}) \cup \beta_U\,.$$
By the inductive hypothesis, $L_j[\beta^\prime \cup S(\beta^\prime)] = \top$.
Since $\beta_U \subseteq \beta^\prime$, we must have $L_j[\beta_U][\beta^\prime \cup S(\beta^\prime)] = \top$, by the restrictive closure of $\PS$. Therefore $L_i[\beta^\prime \cup S(\beta^\prime)] = \top$.
As $\beta$ and $\beta^\prime$ agree on all universal blocks left of $U$ with respect to $\Q$, $S(\beta)$ and $S(\beta^\prime)$ agree on all existential blocks left of $U$, by definition of $\exists$-strategy.
Then, as $U$ is rightmost in $L_j$, $\beta \cup S(\beta)$ and $\beta^\prime \cup S(\beta^\prime)$ agree on $\vars(L_i)$.
Therefore $L_i[\beta \cup S(\beta)] = \top$.
Thus $S$ models $L_i$, completing the inductive step and the proof.
\end{proof}

We arrive at the following result.

\begin{thm}
\label{thm:pred-proof-system}
If $\PS$ is a base system, then $\Pred$ is a sound and complete QBF proof system.
\end{thm}
\begin{proof}
Immediate from Lemmata~\ref{lem:Pred-completeness} and~\ref{lem:Pred-soundness}, and the fact that any $\Pred$ refutation is also a semantic $\Pred$ refutation.
\end{proof}

\section{The Size-Cost-Capacity Theorem} %
\label{sec:technique}
In this section, we define the measures `cost' (Subsection~\ref{subsec:cost}) and `capacity' (Subsection~\ref{subsec:capacity}), then state and prove our central result, the Size-Cost-Capacity Theorem (Subsection~\ref{subsec:SCC}).
In Subsection~\ref{subsec:SCC}, we also formalise strategy extraction in $\Pred$ and prove its correctness.
Throughout this section we assume that $\PS$ is a base system.

\subsection{Cost}
\label{subsec:cost}

Recall that in Section~\ref{sec:equality} we proved the hardness of the equality formulas $\eq(n)$ in $\QURes$ by appealing to the minimum cardinality of the range of a winning strategy.
For $\eq(n)$, the minimum cardinality is easy to compute because the winning strategy per instance is unique, and must contain all possible $2^n$ responses for the single universal block.   

In order to generalise that proof method to arbitrary instances in $\Pred$, we require a more sophisticated measure that accounts for the multiple responses collected during round-based strategy extraction in general.
Fit for this purpose, we define a measure called \emph{cost}.
The cost of a false QBF is the minimum, over all winning strategies, of the largest number of responses for a single universal block.

\begin{defi}[cost]
\label{def:cost}
Let $\Phi = \QBF{\forall U_1 \exists E_1 \cdots \forall U_n \exists E_n}{\phi}$ be a false QBF.
Further, for each winning $\forall$-strategy $S$ for $\Phi$ and each $i \in [n]$, let $S_i$ be the function that maps each $\alpha \in \langle \vars_\exists(\Phi)\rangle$ to the projection of $S(\alpha)$ to $U_i$, and let $\cost(S) = \max\{|\rng(S_i)| : i \in [n]\}$.
The \emph{cost} of $\Phi$ is
$$\cost(\Phi) = \min\{\cost(S) : \mbox{$S$ is a winning $\forall$-strategy for $\Phi$}\}.$$
\end{defi}
The cost of $\eq(n)$ is simple to compute.
There is only one winning strategy and only one universal block exhibiting $2^n$ responses; hence $\cost(\eq(n)) = 2^n$.
\begin{prop}
\label{prop:equality-cost}
The cost of the $n^{\text{\scriptsize}th}$ equality formula is $2^n$.
\end{prop}

\subsection{Capacity}
\label{subsec:capacity}

The principal notion of strategy extraction is that the $\forall$-player's response (for any given round) can be read off from a suitable restriction of the refutation; as we described in Section~\ref{sec:equality}, in $\QURes$ the response can be determined from the penultimate clause.
In the general setting of $\Pred$, we seek a method of determining the response that does not depend upon the particulars of $\PS$.
We introduce the concept of a \emph{response map} for this purpose.

For base system $\PS$ and a line $L$ in $\mathcal{L}_\PS^\Q$, the \emph{rightmost block} of $L$ is the rightmost block $Z$ of $\Q$ for which $\vars(L)$ contains a variable in $Z$.

\begin{defi}[response map]
\label{def:response-map}
Let $\PS$ be a base system, let $L$ be a line in $\mathcal{L}_\PS^\Q$ whose rightmost block $U$ is universal, and let $X := \vars(L) \setminus U$.
A \emph{response map} for $L$ with respect to $\Q$ is a function $\mathcal{R} : \langle X \rangle \rightarrow \langle U \rangle$ satisfying the following for each $\alpha \in \langle X \rangle$:
$$\mbox{If $L[\alpha]$ is not a tautology, then $\R(\alpha)$ falsifies $L[\alpha]$.}$$
\end{defi}

We call a line $L \in \mathcal{L}^\Q_\PS$ \emph{reducible} if its rightmost block is universally quantified.

\begin{defi}[response map set]
\label{def:response-map-set}
Given a semantic $\Pred$ refutation $\pi$ of a QBF $\QBF{\Q}{\phi}$ whose reducible lines are $L_1, \dots, L_k$, a \emph{response map set} for $\pi$ is a set $\{\R_1, \dots ,\R_k\}$ in which each $\R_i$ is a response map for $L_i$ with respect to $\Q$.
\end{defi}


\begin{exa}
\label{ex:QU-response-map}
We can illustrate Definitions~\ref{def:response-map} and~\ref{def:response-map-set} with an example.
Let $\Phi$ be a QBF, and let $\pi$ be a $\QURes$ refutation of $\Phi$ whose reducible clauses are $C_1, \dots ,C_k$.
Consider some particular reducible clause $C_i$ in $\pi$ whose rightmost block is $U_i$, and let $X_i = \vars(C_i) \setminus U_i$.
Further, let $\beta_i$ be any total assignment to $U_i$ that falsifies the $U_i$-literals of $C_i$.
The function $\R_i$ mapping each $\alpha_i \in \langle X_i \rangle$ to $\beta_i$ is a response map for $C_i$, and $\{\R_1 ,\dots ,\R_k\}$ is a response map set for $\pi$.
To see this, it is sufficient to observe that, for $\alpha_i \in \langle X_i \rangle$, if $\alpha_i$ satisfies $C_i$, then $C_i[\alpha_i]$ is a tautology; otherwise, $\beta_i$ falsifies $C[\alpha_i]$.
\end{exa}

This example demonstrates that the lines in resolution -- clauses -- admit very simple response maps that are, in fact, constant functions.
Indeed, given a $\QURes$ refutation we can construct a response map set which needs only one response per reducible line.
In the sequel, given a refutation we shall be interested in response map sets needing the fewest responses; formally, response maps whose ranges have minimal cardinality.
This minimal cardinality is determined by the \emph{capacity} of a refutation.

\begin{defi}[capacity]
\label{def:capacity}
Let $\pi$ be a $\Pred$ refutation of a QCNF $\QBF{\Q}{\phi}$.
The \emph{capacity} of a response map set $\{\R_1, \dots \R_k\}$ for $\pi$ is $\max_{i \in [k]}\{|\rng(\R_i)|\}$.
The \emph{capacity} of $\pi$ is the minimal capacity of any response map set for $\pi$.
\end{defi}

Intuitively, this definition ensures that, given any $\Pred$ refutation $\pi$, there is some response map set for $\pi$ in which each line uses at most $\capa(\pi)$ responses.
Since $\QURes$ refutations have response map sets in which each line uses exactly one response, the capacity of any $\QURes$ refutation is $1$.

\begin{prop}
\label{prop:QURes-capacity}
Every $\QURes$ refutation has capacity equal to $1$.
\end{prop}

\begin{proof}
Let $\pi$ be a $\QURes$ refutation.
The response map $\{\R_1, \dots ,\R_k\}$ given in Example~\ref{ex:QU-response-map} satisfies $|\rng(\R_i)| = 1$ for each $i$.
The capacity of any response map is at least $1$; hence $\capa(\pi) = 1$.
\end{proof}


\subsection{Strategy extraction and Size-Cost-Capacity}
\label{subsec:SCC}

We begin this subsection with a definition of strategy extraction for $\Pred$.
Given a semantic $\Pred$ refutation and response map for it, we may obtain a winning $\forall$-strategy by round-by-round restriction of the refutation, whereby the universal response for a given round is obtained by querying the response map on the first suitable reducible line.

\begin{defi}[round-based strategy extraction]
Let $\pi$ be a semantic $\Pred$ refutation of a QBF $\Phi = \QBF{\exists E_1 \forall U_1 \cdots \exists E_n \forall U_n}{\phi}$ with reducible lines $L_1, \dots ,L_k$ and response map set $\R = \{\R_1, \dots ,\R_k\}$.
Further, for each  $\alpha \in \langle \vars_\exists(\Phi) \rangle$ and $i \in [n]$:
\begin{itemize}
\item let $\sigma^\alpha_0$ be the empty assignment;
\item let $\alpha_i$ be the projection of $\alpha$ to $E_i$;
\item let $F^\alpha_i$ be the first line in $\pi$ such that $\vars(F^\alpha_i) \subseteq \bigcup^i_{j=1}(E_j \cup U_j)$  and $F^\alpha_i[\sigma^\alpha_{i-1} \cup \alpha_i]$ is non-tautological; 
\item if $F^\alpha_i$ is some $L_{j^\alpha_i}$, then let $\beta^\alpha_i := \R_{j^\alpha_i}(\sigma^\alpha_{i-1} \cup \alpha_i)$, otherwise let $\beta^\alpha_i := \R_{j^U_i}(\sigma^\alpha_{i-1} \cup \alpha_i)$, where $j^U_i$ is the minimal index for which $U_i$ is the rightmost block of $L_{j^U_i}$.
\item let $\sigma^\alpha_i := \sigma^\alpha_{i-1} \cup \alpha_i \cup \beta^\alpha_i$.
\end{itemize}

The \emph{extracted strategy} for $\pi$ with respect to $\R$ is
\begin{equation*}
\begin{tabular}{rcl}
$\St_\R(\pi):\langle \vars_\exists(\Phi) \rangle$&$\rightarrow$&$\langle \vars_\forall(\Phi) \rangle$\\
$\alpha$&$\mapsto$&$\bigcup^n_{i=1}\beta^\alpha_i\,.$\\
\end{tabular}
\end{equation*}
\end{defi}

It is clear that the extracted strategy is well-defined, since each $F^\alpha_i$ exists; in particular, for any $i$ and $\alpha$, $\vars(\bot) \subseteq \bigcup^i_{j=1}(E_j \cup U_j)$ and $\bot[\sigma^\alpha_{i-1} \cup \alpha_i]$ is non-tautological, where $\bot$ is the final line in $\pi$.
(Moreover, each natural number $j^U_i$ exists, since we can assume without loss of generality that each universal variable appears in some axiom of the refutation.)
To assist the proof that the extracted strategy $\St_\R(\pi)$ is indeed winning, we first prove two useful propositions which demonstrate that semantic $\Pred$ refutations are closed under certain restrictions.
These two lemmata are essentially the $\Pred$ analogues of Propositions~\ref{prop:QURes-existential-restriction} and ~\ref{prop:QURes-universal-restriction}.  

\begin{prop}
\label{prop:existential-restriction}
Let $\pi$ be a semantic $\Pred$ refutation of a QBF $\Phi$ and let $\alpha$ be a partial assignment to $\vars_\exists(\Phi)$.
Then $\pi[\alpha]$ is a semantic $\Pred$ refutation of $\Phi[\alpha]$.
\end{prop}

\begin{proof}
Let $\pi = L_1, \dots, L_m$.
Each $L_i$ is derived either as an axiom, as a semantic consequence of preceding lines, or from a preceding line by semantic universal reduction.
If $L_i$ is an axiom, then $L_i[\alpha]$ may be derived as an axiom from $\Phi[\alpha]$. 
If $L_i$ is a semantic consequence of the preceding lines $L_{i_1}, \dots ,L_{i_k}$, then $L_i[\alpha]$ is a semantic consequence of the preceding lines $L_{i_1}[\alpha], \dots ,L_{i_k}[\alpha]$, by the restrictive closure of $\PS$.
If $L_i$ was derived from a preceding line $L_j$ by semantic universal reduction, then $L_i \equiv L_j[\beta]$, where $\beta$ is a universal assignment.
Since $\vars(\alpha)$ and $\vars(\beta)$ are disjoint, $L_i[\alpha]$ is semantically equivalent to $L_j[\beta][\alpha]$, and hence also to $L_j[\alpha][\beta]$ (again, by the restrictive closure of $\PS$), and may therefore be derived from $L_j[\alpha]$ by semantic universal reduction.
It follows that $\pi[\alpha] = L_1[\alpha], \dots, L_m[\alpha]$ is a semantic refutation of $\Phi[\alpha]$.
\end{proof}

\begin{prop}
\label{prop:universal-restriction}
Let $\pi$ be a semantic $\Pred$ refutation of a QBF $\Phi$ whose first block $U$ is universal, let $L$ be the first non-tautological line appearing in $\pi$ for which $\vars(L) \subseteq U$, and let $\beta \in \langle U \rangle$.
If $L[\beta] = \bot$ then $\pi[\beta]$ is a semantic $\Pred$ refutation of $\Phi[\beta]$.
\end{prop}

\begin{proof}
Let $\pi = L_1, \dots, L_m$, and suppose that $L[\beta] = \bot$.
As in the proof of the preceding proposition, if $L_i$ is an axiom or a semantic consequence of the preceding lines $L_{i_1}, \dots ,L_{i_k}$, then $L_i[\beta]$ may be derived as an axiom from $\Phi[\beta]$ or as a semantic consequence of the preceding lines $L_{i_1}[\beta], \dots ,L_{i_k}[\beta]$.
If $L_i$ was derived from a preceding line $L_j$ by semantic universal reduction, then $L_i \equiv L_j[\beta^\prime]$, where $\beta^\prime$ is a partial assignment to some universal block $U^\prime$.
We consider two cases:

\emph{Case 1.} Suppose that $U^\prime = U$.
Then $L_j$ contains no variables right of $U$.
Moreover, $L_i$ cannot be $L$, since this would imply that $L$ was derived by semantic universal reduction from a tautology, and is hence itself a tautology by restrictive closure of $\PS$.
Therefore $L_i$ occurs either before or after $L$ in $\pi$.
If $L_i$ appears before $L$, then $L_i$ is a tautology, in which case $L_i[\beta]$ is also a tautology (by restrictive closure of $\PS$) and can be derived as a semantic consequence from any set of lines.
Otherwise $L_i$ appears after $L$, and then $L_i[\beta]$ can be derived as a semantic consequence of $L[\beta] = \bot$.

\emph{Case 2.} Suppose instead that $U^\prime \neq U$.
Then $\vars(\beta)$ and $\vars(\beta^\prime)$ are disjoint.
As in the proof of Proposition~\ref{prop:existential-restriction}, it follows that $L_i[\beta]$ can be derived by semantic universal reduction from $L_j[\beta]$.

It follows that $\pi[\beta] = L_1[\beta], \dots, L_m[\beta]$ is a semantic refutation of $\Phi[\beta]$.
\end{proof}

The preceding propositions allow us to prove fairly easily that the extracted strategy is indeed winning, regardless of the choice of response map.
This is the subject of the following lemma.

\begin{lem}
\label{lem:Pred-strategy-extraction}
Let $\R$ be a response map set for a semantic $\Pred$ refutation $\pi$ of a QBF $\Phi$.
Then the extracted strategy for $\pi$ with respect to $\R$ is a winning $\forall$-strategy for~$\Phi$.
\end{lem}

\begin{proof}
Let $\Phi = \QBF{\exists E_1 \forall U_1 \cdots \exists E_n \forall U_n}{\phi}$, and let $\St$ be the extracted strategy for $\pi$ with respect to $\R$.
We need to show that, for each $\alpha \in \langle \vars_\exists(\Phi)\rangle$, $\alpha \cup \St(\alpha)$ falsifies $\phi$; that is, $\phi[\sigma^\alpha_n] = \bot$.
In fact, we prove by induction on $i = 0, 1, \dots, n$ that $\pi[\sigma^\alpha_i]$ is a semantic $\Pred$ refutation of $\Phi[\sigma^\alpha_i]$.
Hence at the final step $i = n$, by the soundness of $\Pred$ (Lemma~\ref{lem:Pred-soundness}) we prove that $\Phi[\sigma^\alpha_n]$ is false, and this indeed implies $\phi[\sigma^\alpha_n] = \bot$ since $\sigma^\alpha_n$ is a total assignment to $\vars(\Phi)$.

Let $\alpha \in \langle \vars_\exists(\Phi)\rangle$.
As $\sigma^\alpha_0$ is the empty assignment, the base case $i = 0$ is established from the Lemma statement.
For the inductive step, let $i \in [n]$ and suppose that $\pi[\sigma^\alpha_{i-1}]$ is a semantic $\Pred$ refutation of $\Phi[\sigma^\alpha_{i-1}]$.
Since $\alpha_i$ is a partial assignment to $\vars_\exists(\Phi)$, $\pi[\sigma^\alpha_{i-1} \cup \alpha_i]$ is a semantic $\Pred$ refutation of $\Phi[\sigma^\alpha_{i-1} \cup \alpha_i]$, by Proposition~\ref{prop:existential-restriction}.
Now, let $F^\alpha_i$ be the first line in $\pi$ such that $\vars(F^\alpha_i) \subseteq \bigcup^i_{j=1}(E_j \cup U_j)$  and $F^\alpha_i[\sigma^\alpha_{i-1} \cup \alpha_i]$ is non-tautological.
Observe that $F^\alpha_i[\sigma^\alpha_{i-1} \cup \alpha_i]$ is the first non-tautological line in $\pi[\sigma^\alpha_{i-1} \cup \alpha_i]$ whose variables are contained in $U_i$.
If $\vars(F^\alpha_i[\sigma^\alpha_{i-1} \cup \alpha_i])$ contains some variable in $U_i$, then $F^\alpha_i$ is indeed some reduction line $L_{j^\alpha_i}$, and $\beta^\alpha_i \in \langle U_i\rangle$ falsifies $L_{j^\alpha_i}[\sigma^\alpha_{i-1} \cup \alpha_i]$ by definition of response map.
Otherwise $\vars(F^\alpha_i[\sigma^\alpha_{i-1} \cup \alpha_i])$ is empty, and $\beta^\alpha_i$ falsifies $L_{j^\alpha_i}[\sigma^\alpha_{i-1} \cup \alpha_i]$ trivially.
In either case, $\pi[\sigma^\alpha_{i-1} \cup \alpha_i \cup \beta^\alpha_i]$ is a semantic $\Pred$ refutation of $\Phi[\sigma^\alpha_{i-1} \cup \alpha_i \cup \beta^\alpha_i]$, by Proposition~\ref{prop:universal-restriction}.
As $\sigma^\alpha_{i-1} \cup \alpha_i \cup \beta^\alpha_i = \sigma^\alpha_i$, this completes the inductive step, and the proof.
\end{proof}

We defined round-based strategy extraction relative to an arbitrary response map.
Hence, we may select one with minimum capacity; that is, we may select one that minimises the maximum number of responses extracted from a single line.
By selecting such a minimal response map, we will therefore limit the capacity for lines in the refutation to contribute multiple responses to the extracted strategy.
In this way we obtain our main result, which relates size, cost and capacity with a refutation-size lower bound.

\begin{thm}
\label{thm:costcap}
Let $\pi$ be a semantic $\Pred$ refutation of a QBF $\Phi$.
Then
$$|\pi| \geq \frac{\cost(\Phi)}{\capa(\pi)}.$$
\end{thm}

\begin{proof}
Let $\Phi = \QBF{\exists E_1 \forall U_1 \cdots \exists E_n \forall U_n}{\phi}$.
Further, let $L_1, \dots ,L_k$ be the reducible lines of $\pi$ and let $\R = \{\R_1, \dots ,\R_k\}$ be a response map set for $\pi$ with the following property:
for each $i \in [k]$, the cardinality of the range of $\R_i$ is minimal; that is, for each response map $\R^\prime_i$ for $L_i$, $|\rng(\R_i)| \leq |\rng(\R^\prime_i)|$.
Further, let $\St$ be the extracted strategy for $\pi$ with respect to $\R$.
By Lemma~\ref{lem:Pred-strategy-extraction}, $\St$ is a winning $\forall$-strategy for~$\Phi$.
By definition of cost, there exists some universal block $U_i$ such that $|\{\proj(\St(\alpha),U_i) : \alpha \in \langle \vars_\exists(\Phi) \rangle\}|$ is at least $\cost(\Phi)$.
Equivalently, there exists some $i \in [n]$ such that $|\{\beta^\alpha_i : \alpha \in \langle \vars_\exists(\Phi) \rangle\}| \geq \cost(\Phi)$.
Each $\beta^\alpha_i$ is in the range of the response map $\R_{j^\alpha_i}$ for some reducible line $L_{j^\alpha_i}$ in $\pi$.
By definition of capacity, the cardinality of the range of $\R_{i, \alpha}$ is at most $\capa(\pi)$.
Hence we must have $|\pi| \geq \cost(\Phi) / \capa(\pi)$, for otherwise the number of lines in $\pi$ would exceed its size.
\end{proof}

We emphasize that the lower bound given by Theorem~\ref{thm:costcap} applies to \emph{semantic} $\Pred$ refutations.
Therefore, any result obtained as an application of our technique is by definition a genuine lower bound.

Since $\QURes$ derivations have unit capacity (Proposition~\ref{prop:QURes-capacity}), the Size-Cost-Capacity Theorem tells us that \emph{cost alone} is an absolute refutation-size lower bound in that system.
\begin{cor}
\label{cor:QURes-bound}
Let $\pi$ be a $\QURes$ refutation of a QBF $\Phi$. Then $|\pi| \geq \cost(\Phi)$.
\end{cor}

For a straightforward application of our technique, we could use Corollary~\ref{cor:QURes-bound} to prove the hardness of the equality formulas in $\QURes$ (Theorem~\ref{thm:QURes-equality}), as a direct consequence of their exponential cost (Proposition~\ref{prop:equality-cost}), as opposed to the direct method of Section~\ref{sec:equality}.

\section{Capacity bounds} %
\label{sec:cap-bounds}
In this section, we demonstrate that this lower bound technique is widely applicable to QBF proof systems by showing upper bounds on the capacity of proofs in the QBF versions of two commonly studied propositional proof systems: Cutting Planes (Subsection \ref{subsec:CPred}) and Polynomial Calculus with Resolution (Subsection \ref{subsec:PCRred}). These proof systems represent two distinct approaches to propositional proof systems, via integer linear programming and algebraic methods respectively. Both proof systems are known to p-simulate resolution, and similarly the QBF proof systems obtained with the addition of the $\forall$-reduction rule both p-simulate $\QURes$. Our capacity upper bound for $\PCRred$ proofs is particularly noteworthy as it is not constant, but depends on the size of the proof. We conclude this section with an example of the limits of this technique, namely a $\Fred$ proof with large capacity. 

\subsection{Cutting Planes}
\label{subsec:CPred}

The first proof system we analyse is Cutting Planes \cite{CookCT87} and its extension to QBFs, $\CPred$ \cite{BeyersdorffCMS18CP}. Inspired by integer linear programming, Cutting Planes translates a CNF into an equivalent system of linear inequalities, and from these derives the trivial falsehood $0 \geq 1$. Replacing the axiom rules by any unsatisfiable set of inequalities, Cutting Planes is in fact complete for any set of linear inequalities without integer solutions, and is implicationally complete with the inclusion of the trivial truth $0 \geq -1$ as an axiom. Cutting Planes is therefore a base system. As we focus only on Cutting Planes as a proof system for unsatisfiable CNFs, for which the addition of $0 \geq -1$ has no effect on proof size, our definition of Cutting Planes omits this axiom, in keeping with previous work.

Cutting Planes has two inference rules: the linear combination rule and the division rule. The linear combination rule infers from two inequalities, some linear combination of these inequalities with non-negative integer coefficients. The division rule allows division by any integer $c > 0$ if $c$ divides the coefficient of each variable; note that $c$ need not divide the constant term in the inequality.

\begin{defi}[Cutting Planes \cite{CookCT87}]
A line $L$ in a Cutting Planes {$\CP$} proof is a linear inequality $a_1 x_1 + \dots a_n x_n \geq A$ where $\vars (L) = \{ x_1 , \dots, x_n \}$ and $a_1 , \dots , a_n , A \in \mathbb{Z}$.\footnote{For convenience and clarity, we may refer to lines in $\mathcal{L}_\CP$ and $\mathcal{L}^\Q_\CP$ using equivalent linear inequalities not in this precise form. Similarly, the result of any $\forall$-reduction is expressed as a line of this form.}

A Cutting Planes derivation of a line $L \in \mathcal{L}_\mathsf{CP}$ from a CNF $\phi$ consists of a sequence of lines $L_1 , \dots , L_m$ where $L_m = L$ and each line $L_i \in \mathcal{L}_{\mathsf{CP}}$ is either an instance of an axiom rule, or is derived from the previous lines by an inference rule (Figure~\ref{fig:CPRules}). A $\mathsf{CP}$ refutation of $\phi$ is a derivation of $0 \geq 1$ from $\phi$.
\begin{figure}[h]
\CPDerivationRules
\caption{The rules of the Cutting Planes proof system ($\CP$)}
\label{fig:CPRules}
\end{figure}
\end{defi}

It is straightforward to see that Cutting Planes $p$-simulates resolution.\footnote{A proof system $\mathsf{P}$ \emph{p-simulates} a system $\mathsf{Q}$ if each $\mathsf{Q}$-proof can be transformed in polynomial time into a $\mathsf{P}$-proof of the same formula \cite{CookR79}.}
Indeed, it is strictly stronger than resolution, as there are short $\CP$ proofs of the pigeonhole principle formulas, which are known to require large proofs in resolution \cite{Haken85}. The same argument shows that $\CPred$ is exponentially stronger than $\QURes$.

Despite the apparent strength of $\CPred$ compared to $\QURes$, any proof in $\CPred$ still only has unit capacity. This comes about as the left hand side of any inequality $L$ is simply a linear combination of variables. We can therefore construct a constant response map defined by the coefficients of the variables of the rightmost universal block, with the aim of minimising the left hand side of the inequality.

\begin{prop}
For every $\CPred$ derivation $\pi$, $\capa (\pi) = 1$.
\label{prop:CP-cap}
\end{prop}

\begin{proof}
Consider any $L \in \mathcal{L}^\Q_\CP$ with rightmost block $U$ universally quantified, and with $X = \vars(L) \setminus U$. Then $L$ is of the form $\sum_{x \in X} a_x x + \sum_{u \in U} b_u u \geq c$, where $a_x, b_u, c \in \mathbb{Z}$ for all $x \in X, u \in U$.

Define the assignment $\beta_L \in \langle U \rangle$ by assigning each variable $u$ to 0 if $b_u \geq 0$, and to 1 otherwise. It is clear that this response minimises the value of $\sum_{u \in U} b_u u$, and thus minimises the left hand side of the inequality in the line $L[\alpha]$ for any $\alpha \in \langle X \rangle$. If any response falsifies $L[\alpha]$, a response which minimises this term must do so, so the map $\R_L : \langle X \rangle \to \langle U \rangle$ defined by $\R_L (\alpha) = \beta_L$ is a response map for $L$.

For any $\CPred$ proof $\pi$, the set $\{\R_L : \mbox{$L$ is a reducible line of $\pi$}\}$ is a response map set for $\pi$, witnessing that $\capa(\pi) = 1$.
\end{proof}

Having established that any $\CPred$ proof has capacity 1, we can apply the Size-Cost-Capacity Theorem (Theorem~\ref{thm:costcap}). This yields a lower bound on the size of $\CPred$ proofs of a QBF $\Phi$ determined solely by the cost of $\Phi$.
\begin{cor}
\label{cor:CP-bound}
Let $\pi$ be a $\CPred$ refutation of a QBF $\Phi$. Then $|\pi| \geq \cost(\Phi)$.
\end{cor}

Hence, even in the stronger system of $\CPred$, we still have a straightforward proof that refutations of the equality formulas require size $2^n$ by looking at the cost of the formulas and using Size-Cost-Capacity.

\subsection{Polynomial Calculus}
\label{subsec:PCRred}
Polynomial Calculus \cite{CleggEI96} presents an algebraic approach to proving unsatisfiability. A CNF $\phi$ is translated into a set of polynomials over an arbitrary field $\mathbb{F}$, such that any assignment where all polynomials evaluate to zero corresponds to a satisfying assignment for $\phi$, and vice versa.
The field $\mathbb{F}$ is often assumed to be $\mathbb{Q}$, but since all axioms and any solutions use only the constants 0 and 1, Polynomial Calculus (as a proof system for CNFs) can be defined over any field $\mathbb{F}$.

Formally, Polynomial Calculus works with polynomial equations where the right hand side is 0. A Polynomial Calculus refutation of a set of polynomials is a derivation of the equation $1 = 0$, which is enough to show that the set of polynomials has no common solution.
The inference rules permit deriving any linear combination of two previous lines, or multiplying any line by a single variable. 

As a propositional proof system, the axioms of Polynomial Calculus are polynomials equivalent to each clause in a CNF. Given a clause $C$ in a CNF, the corresponding polynomial axiom is $\prod_{l \in C} V(l) = 0$, where $V(x) = x$ and $V(\neg x) = (1-x)$.\footnote{Assigning the algebraic variable $x$ to 0 is therefore equivalent to assigning the corresponding Boolean variable to 1, and vice versa. This `swapping' of truth values is a common convention for algebraic proof systems.} We also include the axioms $x^2 - x = 0$ to ensure only Boolean solutions are possible. With these Boolean axioms added, Polynomial Calculus is implicationally complete \cite{CleggEI96,FilmusLMNV13}, ensuring that Polynomial Calculus is a base system.

Proof size in Polynomial Calculus is measured by the number of \emph{monomials} in the lines of the proof. By this measure of proof size, Polynomial Calculus cannot even simulate resolution, as the clause $\neg x_1 \vee \dots \vee \neg x_n$ would translate to $(1-x_1) \dots (1 - x_n) = 0$, which contains $2^n$ monomials. As a result of this issue, a modification of Polynomial Calculus was introduced in \cite{AlekhnovichBRW02}, using variables $\bar{x}$ representing $\neg x$.
The inference rules remain the same, but we add the axioms $x + \bar{x} - 1 = 0$ for each $x \in \vars(\phi)$ to ensure that $x$ and $\bar{x}$ take opposite values.

\begin{defi}[Polynomial Calculus with Resolution \cite{AlekhnovichBRW02}]
Fix an arbitrary field $\mathbb{F}$. Given a CNF $\phi$, lines in a $\PCR$ derivation from $\phi$ are polynomials in the variables $\{ x, \bar{x} : x \in \vars(\phi)\}$. A $\PCR$ derivation of a line $L \in \mathcal{L}_{\mathsf{PCR}}$ from a CNF $\phi$ is a sequence of lines $L_1 , \dots , L_m$ in $\mathcal{L}_{\mathsf{PCR}}$ such that $L_m = L$, and each $L_i$ is an instance of an axiom rule, or derived from previous lines by one of the inference rules (Figure~\ref{fig:PCRRules}). A $\PCR$ refutation of $\phi$ is a derivation of the line $1 = 0$.

\begin{figure}[h]
\PCRDerivationRules
\caption{The rules of Polynomial Calculus with Resolution ($\PCR$)}
\label{fig:PCRRules}
\end{figure}
\end{defi}

As Polynomial Calculus with Resolution is clearly at least as strong as Polynomial Calculus (and in fact is strictly stronger), we focus here only on the version with Resolution. The capacity upper bounds shown, and consequent proof size lower bounds, all hold for Polynomial Calculus as well.

In contrast to $\QURes$ and $\CPred$, not all proofs in $\PCRred$ have unit capacity. For a simple example, consider the line $x(1-u) + (1-x)u = 0$. This polynomial clearly evaluates to 0 if and only if $x = u$, so the only winning response for $u$ is to play $u = 1-x$. The unique response map $\R$ for this line therefore has $\rng(\R) = 2$, so if a $\PCRred$ proof $\pi$ contains such a line then $\capa (\pi) \geq 2$.

Given a QBF $\Phi$ with $n$ variables in the universal block $U$, it is possible to construct a line in $\mathcal{L}_{\PCR}^{\Q}$ which requires a response map with an exponential size range.  However, as previously mentioned, the size of a proof in $\PCRred$ is not measured by the number of lines in the proof, but by the number of monomials in the proof. To provide a suitable capacity upper bound, it is thus sufficient to upper bound the size of a response set for a line by the number of monomials in that line. That is, we show that any proof with large capacity must contain a large line, and so the proof itself is large.

\begin{prop}
 If $\pi$ is a $\PCRred$ proof where each line contains at most $M$ monomials, then $\capa (\pi) \leq M$.
 \label{prop:PC-capacity}
\end{prop}

The key observation used in proving this bound is that, rather than considering the specific assignment to the universal variables, it is sufficient to consider the value the response gives to the monomials of a line $L \in \mathcal{L}_{\PCR}^{\Q}$. If two assignments $\beta, \beta' \in \langle U \rangle$ set the same monomials to zero and retain the same monomials, only one of these assignments need be in the range of a response map as $L[\beta] = L[\beta']$ and so $\beta$ and $\beta'$ falsify $L[\alpha]$ for precisely the same $\alpha \in \langle X \rangle$.

We therefore consider the set of different assignments to these monomials needed, as a subset of ${\{0,1 \}}^M$. We construct a response map for $L$ such that the responses, when considered in this way, form a linearly independent subset of ${\{ 0,1 \}}^M$, providing the upper bound.

\begin{proof}[Proof of Proposition~\ref{prop:PC-capacity}]
To upper bound $\capa (\pi)$, it is enough to find, for any reducible line $L \in \pi$, a response map $\R_L$ such that $|\rng (\R_L)| \leq M$. To this end, fix such a line $L \in \pi$ with at most $M$ monomials. Let $U$ to be the rightmost universal block present in  $L$, and $X := \vars(L) \setminus U$, and write $L$ as
 $$\sum_{j=1}^N f_j v_j = 0$$
 where each $f_j$ is a polynomial (\emph{not} necessarily a single monomial) in the variables of $X$ and the $v_j$ are distinct monomials in the variables $U$. Denote by $f_j [\alpha]$ and $v_j[\beta]$ the values (in $\mathbb{F}$) obtained by evaluating $f_j$, respectively $v_j$, according to the assignments $\alpha \in \langle X \rangle$, respectively $\beta \in \langle U \rangle$.
 
 Observe that since $L$ contains at most $M$ monomials, $N \leq M$. We now construct a response map $\mathcal{R}_L: \langle X \rangle \to \langle U \rangle$ for which $|\rng(\R_L)| \leq N$.  The response map set for $\pi$ consisting of $\R_L$ for every reducible line $L \in \pi$ demonstrates that $\capa(\pi) \leq M$.
 
 First, enumerate the elements of $\langle X \rangle$ as $\langle X \rangle = \{\alpha_1 , \dots, \alpha_m\}$. We construct a sequence of functions $\mathcal{R}_L^i : \{ \alpha_1 , \dots, \alpha_i\} \to \langle U \rangle$ such that $\mathcal{R}_L^0$ is the empty function, $\mathcal{R}_L^{i}$ extends $\mathcal{R}_L^{i-1}$, and $\mathcal{R}_L^m = \R_L$ is a complete response map for $L$. Moreover, for each $0 \leq i \leq m$, $|\rng (\mathcal{R}_L^i)| \leq N$, in particular $|\rng (\mathcal{R}_L^m)| \leq N$.
 
 The construction of the function $\mathcal{R}_L^m$ involves, in effect, going through each of the assignments $\alpha \in \langle X \rangle$ in turn and choosing the response of $\mathcal{R}_L^m$ on $\alpha$. At each stage, if there is a suitable response that has been chosen before, we choose it again as we are aiming to minimise $|\rng(\R_L)|$. If there is no suitable previously chosen response, we choose a suitable response and then show that by evaluating the monomials $v_j$ on the responses, there is an injection from the set of chosen responses into a linearly independent subset of $\mathbb{F}^N$.
 
If $L$ is a tautology, then for any $\alpha \in \langle X \rangle$, $L[\alpha]$ cannot be falsified, and so any constant function $\R_L : \langle X \rangle \to \langle U \rangle$ is a response map with range 1. We therefore assume without loss of generality that $L$ is not a tautology, and in particular that $L[\alpha_1]$ is falsifiable.
 
 
 
 Let $R_i := \rng(\mathcal{R}_L^i)$. For any $\beta \in \langle U \rangle$, denote by $\vec{v}[\beta]$ the vector $(v_1 [\beta] , \dots, v_N [\beta]) \in {\{ 0,1 \}}^N$. We now define a response $\beta_i \in \langle U \rangle$ which extends $\R_L^{i-1}$ to $\R_L^i$ such that $\R_L^i$ can be extended to a response map for $L$, i.e.\ whenever $L[\alpha_i]$ is falsifiable, $L[\alpha_i \cup \beta_i] = \bot$.
 
 Furthermore, we show by induction on $i$ that the set $V_i = \{ \vec{v}[\beta] : \beta \in R_i\}$ is linearly independent (as a subset of $\mathbb{F}^N$).  Since $V_i \subseteq {\{0,1\}}^N$, and $|V_i| = |R_i|$, this provides the upper bound on $|R_i|$.
 
 The empty function $\mathcal{R}_L^0$ ensures that $R_0 = \varnothing$, which is linearly independent. Given a function $\mathcal{R}_L^{i-1}$, we define $\mathcal{R}_L^i$ as follows:
 \begin{itemize}
 
  \item If $L[\alpha_i]$ is a tautology, i.e. $\sum_{j=1}^N f_j [\alpha_i] v_j [\beta]= 0$ for all $\beta \in \langle U \rangle$, a response map can map $\alpha_i$ to any response in $\langle U \rangle$. Since $i \neq 1$, $R_{i-1}$ is non-empty, so define $\R_L^i (\alpha_i) = \R_L^{i-1} (\alpha_1)$, whence $R_i = R_{i-1}$.  As $V_i = V_{i-1}$ and $V_{i-1}$ is linearly independent, $V_i$ is also linearly independent.

   \item If $L[\alpha_i \cup \beta] = \bot$ for some $\beta \in R_{i-1}$, then define $\R_L^i (\alpha_i) = \beta$. Now, as previously, $R_i = R_{i-1}$ and $V_i$ is linearly independent as $V_i = V_{i-1}$.
   
   \item Else $L[\alpha_i]$ is not a tautology, but $L[\alpha_i \cup \beta] = \top$ for all $\beta \in R_{i-1}$. Suppose $R_{i-1} = \{ \beta_1 , \dots , \beta_k \}$, then there exists some $\beta_{k+1}$ such that $L[\alpha_i \cup \beta_{k+1}] = \bot$, since $L[\alpha_i]$ is not a tautology. Define $\R_L^i (\alpha_i) = \beta_{k+1}$, with $R_i = \{ \beta_1 , \dots, \beta_k , \beta_{k+1}\}$. We now need only show that $V_i = \{ \vec{v}[\beta_j]: 1 \leq j \leq k+1 \}$ is linearly independent.
   
For any $1 \leq l \leq k$, $L[\alpha_i][\beta_l] \neq L[\alpha_i][ \beta_{k+1}]$, so $\vec{v}[\beta_{k+1}] \neq \vec{v}[\beta_l]$. Suppose there is some linear dependence relation on $V_i$. Since $V_{i-1}$ is linearly independent, $\vec{v}[\beta_{k+1}]$ must have a non-zero coefficient in any such linear combination, hence there are constants $c_1 , \dots ,c_k \in \mathbb{F}$ such that $\sum_{t=1}^k c_t \vec{v}[\beta_t] = \vec{v} [\beta_{k+1}]$. If such constants exist, we can use the same constants to construct a linear combination of the $\sum_{j=1}^N f_j [\alpha_i] v_j [\beta_t]$, by assumption all equal to zero, summing to $\sum_{j=1}^N f_j [\alpha_i] v_j [\beta_{k+1}]$, which by choice of $\beta_{k+1}$ is non-zero.
   $$ 0 = \sum_{t=1}^k { c_t  \sum_{j=1}^N f_j [\alpha_i] v_j [\beta_t] } = \sum_{j=1}^N f_j [\alpha_i] \sum_{t=1}^k c_t v_j[\beta_t] = \sum_{j=1}^N f_j [\alpha_i] v_j [\beta_{k+1}] \neq 0$$
   From this contradiction, we conclude that the constants $c_t$ do not exist, and thus that $V_i$ is a linearly independent set.
  
 \end{itemize}
 
 The set $V_m$ forms a linearly independent set which is a subset of $\mathbb{F}^N$, so has cardinality at most $N$. By the construction of each $\mathcal{R}_L^i$, $R_i$ cannot contain $\beta \neq \beta'$ with $\vec{v}(\beta) = \vec{v}(\beta')$, so $|R_m| = |V_m|$. Consequently, the map $\mathcal{R}_L = \mathcal{R}_L^m$ satisfies $|\rng (\mathcal{R}_L)| \leq N \leq M$. We therefore obtain the response map set $\{ \R_L : \mbox{$L \in \pi$ is reducible}\}$ for $\pi$ and conclude that $\capa (\pi) \leq M$.
\end{proof}

The effect of this bound is to show that $\PCRred$ proofs with large capacity also have large size, as they must contain lines with a large number of monomials. This provides a lower bound for $\PCRred$ proofs of a QBF $\Phi$ based solely on $\cost (\Phi)$, since small proofs also have small capacity.

\begin{cor}
\label{cor:PC-bound}
Let $\pi$ be a $\PCRred$ refutation of a QBF $\Phi$. Then $|\pi| \geq \sqrt{\cost(\Phi)}$.
\end{cor}

\begin{proof}
As the size of $\pi$ is measured by the number of monomials, each line of $\pi$ contains at most $|\pi|$ monomials, and so by Proposition~\ref{prop:PC-capacity}, $\capa (\pi) \leq |\pi|$. Applying Size-Cost-Capacity (Theorem~\ref{thm:costcap}), we conclude that $|\pi| \geq \frac{\scriptsize{\cost}(\Phi)}{|\pi|}$, i.e. $|\pi| \geq \sqrt{\cost(\Phi)}$.
\end{proof}

As for $\QURes$ and $\CPred$, this immediately gives a lower bound of $2^{\Omega(n)}$ for any proof of the equality formulas in $\PCRred$.

\subsection{Proofs with large capacity}

We conclude this section by noting that our technique cannot be applied to some of the more powerful proof systems. These proof systems use lines which are able to concisely express more complex Boolean functions which require large response sets. The example we give is a proof of the equality formulas in the proof system $\Fred$. The $\Fred$ proof system is the strongest proof system we discuss in this paper, and no superpolynomial lower bounds on proof size are known in the propositional system $\Frege$, nor in the QBF proof system $\Fred$.

The $\Frege$ proof system is a `textbook' propositional proof system, in which lines are arbitrary formulas in propositional variables, the constants $\top, \bot$ and the connectives $\wedge , \vee , \neg$.
The rules of a $\Frege$ system comprise a set of axiom schemes and inference rules, which must be implicationally complete \cite{CookR79}; all such systems are equivalent \cite{CookR79, Krajicek95}. 

\begin{prop}
\label{prop:frege}
There is a $\Fred$ refutation of the $n^{\text{\scriptsize}th}$ equality formula $\eq(n)$ with size $O(n)$.
\end{prop}

\begin{proof}
From the lines $x_i \vee u_i \vee \lnot t_i$ and $\neg x_i \vee \neg u_i \vee \lnot t_i$, there is a constant size $\Fred$ derivation of the line $L_i = [(x_i \vee u_i) \wedge (\neg x_i \vee \neg u_i )] \vee \lnot t_i$. Successively applying the resolution rule, which $\Frege$ can simulate with constant size, to the lines $L_i$ and the line $t_1 \vee \dots \vee t_n$ results in the line
$$ L = \bigvee_{i=1}^n \left[ \left( x_i \vee u_i \right) \wedge \left( \neg x_i \vee \neg u_i \right) \right]. $$

Let $K_m$ be the line $\bigvee_{i=1}^m [(x_i \vee u_i) \wedge (\neg x_i \vee \neg u_i)]$, so $K_n = L$ and $K_0 = \bot$. For any $m$, $K_{m-1}$ can be derived from $K_m$ by first $\forall$-reducing to obtain $K_m[u_m/0] \equiv K_{m-1} \vee x_m$ and $K_m[u_m/1] \equiv K_{m-1} \vee \neg x_m$, and then resolving these on $x_m$. 

Deriving each $K_m$ in turn from $K_{m+1}$ provides a linear size refutation of $L$, and hence a linear size refutation of $\eq(n)$.
\end{proof}

The $\Fred$ proof $\pi$ defined in the proof above only uses formulas of depth 3 (i.e.\ there are at most two alternations between $\wedge$ and $\vee$) and so $\pi$ is also a refutation in the more restrictive systems $\AC{0}$-$\Fred$ and $\complexityClassFont{AC^0_3}$-$\Fred$.
By the Size-Cost-Capacity Theorem (Theorem~\ref{thm:costcap}), we know that $\capa (\pi)$ is of exponential size. We can show this directly by considering the reducible line $L$ from $\pi$.

Any winning response for $L$ to an assignment $\alpha \in \langle\{ x_1 , \dots , x_n \}\rangle$ must falsify $(x_ i \vee u_i) \wedge (\neg x_i \vee \neg u_i)$ for each $1 \leq i \leq n$. The unique winning response to $\alpha$ is therefore to play $\beta$ such that $\beta(u_i) = \alpha (x_i)$. Since there are $2^n$ distinct assignments in $\langle\{ x_1 , \dots , x_n \}\rangle$, the range of the unique response map for $L$ is $\langle \{ u_1 , \dots , u_n \} \rangle$, which has size $2^n$, despite $L$ being a $\Fred$ line of size polynomial in $n$, and so $\capa(\pi) \geq 2^n$.

\section{Randomly generated formulas with large cost} %
\label{sec:random-QBFs}

In the previous section, we saw that Size-Cost-Capacity can be used to simultaneously show lower bounds in many different QBF proof systems simply by examining the cost of QBFs. We now construct a class of randomly generated QBFs, denoted $Q(n,m,c)$ which with high probability are false and have large cost for appropriate values of $m$ and $c$. By showing a lower bound on cost, we immediately obtain lower bounds on proof size in $\QURes$, $\CPred$ and $\PCRred$ for these random QBFs.

The model used for these formulas is a simple extension of previous constructions of random QBFs, and our lower bounds have both theoretical and practical significance. On the theoretical side, the high cost of these formulas demonstrates that our lower bound technique applies to a large number of QBFs, and not just a few handcrafted instances. For practical QBF solving, it is important to have a large number of instances on which to test a solver. Having identified cost as a simple reason for QBF proof size lower bounds, the utility of generating many QBFs with large cost is clear.

\begin{defi}
\label{def:Qnmc}
For each $1 \leq i \leq n$, let $C_i^1 , \dots , C_i^{cn}$ be distinct clauses picked uniformly at random from the set of clauses containing 1 literal from the set $X_i = \{ x_i^1 , \dots, x_i^m \}$ and 2 literals from $Y_i = \{ y_i^1 , \dots, y_i^n \}$. Define the randomly generated QBF $Q(n,m,c)$ as:
$$ Q(n,m,c) := \exists Y_1 \dots Y_n \forall X_1 \dots X_n \exists t_1 \dots t_n \cdot \bigwedge_{i=1}^n \bigwedge_{j=1}^{cn} \left( \neg t_i \vee C_i^j \right) \wedge \bigvee_{i=1}^n t_i.$$
\end{defi}

Specifying that clauses contain a given number of literals from different sets may seem unusual, especially to readers familiar with random SAT instances, however it is widely used in the study of randomly generated QBFs \cite{ChenI05,CreignouDER15}. If any clause in the matrix of a QBF contains only literals on universal variables, then it is easy to see that the QBF is false, and that all proof systems $\Pred$ have a constant size refutation using only this clause. Specifying that all clauses must contain a given number of literals from different sets of variables avoids this issue by guaranteeing that all clauses contain existential variables. It is natural that we would also expect clauses in a QBF to contain universal variables.



We define the QBFs $\Psi_i := \exists Y_i \forall X_i \cdot \bigwedge_{j=1}^{cn} C_i^j$, which are generated using the same clauses as the (1,2)-QCNF model (Definition \ref{def:12QCNF}). We can easily see that $Q(n,m,c) \equiv \bigvee_{i=1}^n \Psi_i$ is simply a disjunction of $n$ QBFs generated using this commonly studied model for random QBFs. This equivalence allows us to better understand $Q(n,m,c)$ by studying the randomly generated $\Psi_i$. 

We first show that, for suitable values of $m$ and $c$, with probability $1 - o(1)$, all $\Psi_i$ are false, and furthermore for a linear number of the $\Psi_i$, $\cost(\Psi_i) \geq 2$. By observing that a winning strategy for $Q(n,m,c)$ is a winning strategy for all the $\Psi_i$ simultaneously, we can then show a cost lower bound for $Q(n,m,c)$.

We first focus on proving that, with suitable values for $m$ and $c$, $Q(n,m,c)$ is false with high probability. This is equivalent to showing that, with high probability, each of the $\Psi_i$ is false. In each $\Psi_i$, any winning assignment for the $\exists$-player must satisfy an existential literal in every clause. If not, there would be a winning $\forall$-strategy constructed by finding a clause where both existential literals were false, and setting the universal literal in that clause to false. Determining the truth of $\Psi_i$ can therefore be reduced to determining the satisfiability of the 2-SAT problem defined by the existential parts of the clauses $C_i^j$.
We can then use the following result on the satisfiability of random 2-SAT formulas, shown independently by Chv\'{a}tal and Reed \cite{ChvatalR92}, Goerdt \cite{Goerdt96} and de la Vega \cite{Vega98}, to obtain the falsity of the $\Psi_i$. We state it here with a tighter probability lower bound of $1 - o(n^{-1})$ proved by de la Vega in \cite{Vega01j}, which is necessary for our present work.

\begin{thm}[de la Vega \cite{Vega01j}]
Let $\Phi$ be a random 2-SAT formula on $n$ propositional variables containing $cn$ clauses selected uniformly at random. If $c > 1$ then $\Phi$ is unsatisfiable with probability $1 - o \left( n^{-1} \right) $.
\label{thm:2satfalse}
\end{thm}

The following lemma is equivalent to the statement that, with the same bounds on $m$ and $c$, $Q(n,m,c)$ is false with probability $1 - o(1)$.  This is a fairly immediate consequence of Theorem~\ref{thm:2satfalse}; we need only check that it is sufficiently likely that the existential parts of the clauses of the $\Psi_i$ satisfy the conditions of Theorem~\ref{thm:2satfalse}. The possibility of repeating an existential clause many times with different universal literals makes this non-trivial, but the proof is relatively straightforward.

\begin{lem}
For each $1 \leq i \leq n$, let $\psi_i$ be a set of $cn$ clauses picked uniformly at random from the set of clauses containing 1 literal from $X_i = \{ x_i^1 , \dots, x_i^m \}$ and 2 literals from $Y_i = \{ y_i^1 , \dots, y_i^n \}$. If $m \leq \log_2 (n)$ and $c > 1$, then with probability $1 - o(1)$, $\Psi_i := \exists Y_i \forall X_i \cdot \psi_i$ is false for all $1 \leq i \leq n$.
 \label{lem:psisfalse}
\end{lem}

\begin{proof}
For the QBFs $\Psi_i$ to be false, it is sufficient for the 2-SAT problem generated by taking only the existential parts of the clauses to be false, as the universal response need only respond by falsifying the universal literal on some unsatisfied existential clause. However, it is possible that clauses in $\psi_i$ contain the same existential literals and differ only in the universal literal. In order to use Theorem~\ref{thm:2satfalse}, we need to show that there is some constant $k > 1$ such that, for each $i \in [n]$, the clauses of $\psi_i$ contain at least $kn$ distinct pairs of existential literals with high probability.

For each $\psi_i$, there are $4 \binom{n}{2}$ choices for the existential literals of a clause, and $2m \leq 2 \log (n)$ possible universal literals. The total number of possible clauses is therefore at most $4n(n-1)\log (n)$. 

Let $k$ be some constant with $1 < k < c$. To determine the probability of $\psi_i$ containing at least $kn$ distinct clauses in the existential variables, consider successively making $cn$ random choices of clause from the $4n(n-1) \log (n)$ possible clauses. If, on choosing a clause, fewer than $kn$ distinct existential clauses have been chosen, the probability of a randomly chosen clause having existential part distinct from all previously chosen clauses is at least $$\frac{4n(n-1) \log (n) - 2kn \log (n)}{4n(n-1) \log (n)} = 1- \frac{k}{2(n-1)}.$$

Define the selection of a clause to be successful if it either selects a clause with existential part distinct from that of the previous clauses, or if $kn$ distinct existential clauses have already been selected. The probability of any selection being successful is therefore at least $1 - \frac{k}{2(n-1)}$. It is enough to show that if we select $cn$ clauses, with a probability $1 - \frac{k}{2(n-1)}$ of success for each selection, then the probability of fewer than $kn$ successes is $O(e^{-n})$.

The distribution of the random variable $Z$, the total number of successes, is a sum of $cn$ Bernoulli random variables with $p = 1-\frac{k}{2(n-1)}$. Substituting these values into Hoeffding's inequality, we obtain
$$ P(Z \leq kn) \leq \exp \left( -2 \frac{{\left( cn - \frac{kcn}{2(n-1)} - kn \right) }^2}{cn}\right) = \exp \left( -\frac{2{(c-k)}^2}{c}n + O(1)  \right)$$
and so $P(Z > kn) = 1 - \frac{1}{e^{\Omega(n)}} = 1 - o(n^{-1})$.

The probability that a given $\Psi_i$ is false is at least the probability of it containing at least $kn$ distinct existential clauses and the first $kn$ distinct such clauses being unsatisfiable. Given the clauses of $\psi_i$ were chosen uniformly at random, each set of $kn$ existential clauses is equally likely to be chosen, so the probability these clauses are unsatisfiable is $1 - o\left( n^{-1} \right)$, by Theorem~\ref{thm:2satfalse}. The probability of $\Psi_i$ being false is therefore at least $P(Z > kn) \cdot \left(1 - o \left( n^{-1} \right) \right) = 1 - o \left( n^{-1} \right)$.

Finally, the selection of clauses for each $\Psi_i$ is independent of clauses chosen in any other $\Psi_i$, and so the probability of all being false is ${ \left( 1 - o \left( n^{-1} \right) \right) }^n = 1 - o(1)$.
\end{proof}

It remains to show that $\cost (Q(n,m,c))$ is large. Again, we first look at the cost of $\Psi_i$, and observe that, for $m \leq \log_2 (n)$ and $1 < c < 2$, $\cost(\Psi_i) \geq 2$ with probability $1 - o(1)$. Winning responses for $Q(n,m,c)$ are simultaneous winning responses for each of the $\Psi_i$. As many of the $\Psi_i$ require multiple distinct responses, it is reasonable to expect that the number of responses to falsify all of them is large. With a careful choice of the parameters $m$ and $c$, we can indeed force $Q(n,m,c)$ to have a large cost with high probability.

To prove $\cost(\Psi_i) \geq 2$, it is only necessary to show that $\cost(\Psi_i) \neq 1$, i.e. that any winning $\forall$-strategy $S: \langle Y_i \rangle \to \langle X_i \rangle$ for $\Psi_i$ is not constant. If there is a constant winning $\forall$-strategy, say $S(\alpha) = \beta$ for all $\alpha \in \langle Y_i \rangle$, then $\beta$ also constitutes a winning $\forall$-strategy for $\Psi_i' = \forall X_i \exists Y_i \cdot \psi_i$.

\begin{defi}[Chen and Interian \cite{ChenI05}]
\label{def:12QCNF}
A (1,2)-QCNF is a QBF of the form
$ \forall X \exists Y \cdot \phi (X,Y) $
where $X = \{ x_1 , \dots , x_m \}$, $Y = \{ y_1 , \dots , y_n \}$ and $\phi(X,Y)$ is a 3-CNF formula in which each clause contains one universal literal and two existential literals.
\end{defi}

If a winning $\forall$-strategy for $\Psi_i'$ exists, then $\Psi_i'$ is false. However, for $c<2$, $\Psi_i'$ is known to be true with high probability.

\begin{thm}[Creignou et al.\ \cite{CreignouDER15}] 
\label{thm:qsattrue}
Let $\Phi$ be a (1,2)-QCNF in which $\phi (X,Y)$ contains $cn$ clauses picked uniformly at random from the set of all suitable clauses. If $m \leq \log_2 (n)$, and if $c < 2$, then $\Phi$ is true with probability $1 - o(1)$.
\end{thm}

We therefore pick the parameter $c$ to lie between the lower bound from Theorem~\ref{thm:2satfalse} and the upper bound from Theorem~\ref{thm:qsattrue}. From these results, we see that for $1 < c < 2$, $\exists Y \forall X . \psi$ is false, but $\forall X \exists Y . \psi$ is true with high probability. Any constant winning $\forall$-strategy for $\exists Y \forall X . \psi$ is also a winning $\forall$-strategy for $\forall X \exists Y \cdot \psi$, whence the latter is false. This gives us the bound $\cost(\Psi_i) \geq 2$ with high probability.

\begin{lem}
Let $\psi$ be a set of $cn$ clauses picked uniformly at random from the set of clauses containing 1 literal from the set $X = \{ x_1 , \dots, x_m \}$ and 2 literals from $Y = \{ y_1 , \dots, y_n \}$. If $1 < c < 2$ and $m \leq  \log_2 (n) $, then with probability $1 - o(1)$, $\Psi := \exists Y \forall X \cdot \psi$ is false, and $\cost(\Psi) \geq 2$.
\label{lem:nonconststrat}
\end{lem}

\begin{proof}
Observe from the proof of Lemma~\ref{lem:psisfalse} that, as $c > 1$, $\Psi$ is false with probability $1 - o(n^{-1})$ and $\cost (\Psi) \geq 1$.

Suppose $\cost(\Psi) = 1$, then there is some $\beta \in \langle X \rangle$ such that $\beta$ is a winning response for any $\alpha \in \langle Y \rangle$. That is, for any $\alpha \in \langle Y \rangle$, $\psi [\alpha][\beta] = \bot$.
We can use $\beta$ as a winning strategy for $\Psi' = \forall X \exists Y .\psi$, defining $S'(\emptyset) = \beta$. Since $\psi[\beta][\alpha] = \bot$ for all $\alpha \in \langle Y \rangle$, $S'$ is a winning $\forall$-strategy and so $\Psi'$ is false. However since $c < 2$, $\Psi'$ is false with probability $o(1)$ (Theorem~\ref{thm:qsattrue}), and so such a $\beta \in \langle X \rangle$ exists with probability $o(1)$.

The probability that $\Psi$ is false and $\cost(\Psi) \geq 2$ is therefore $1 - o(n^{-1}) - o(1) = 1-o(1)$.
\end{proof}

With this, we can show that a linear number of the $\Psi_i$ require multiple responses with probability $1 - o(1)$. This will be enough to give a large lower bound on $\cost(Q(n,m,c))$.

\begin{lem}
For each $1 \leq i \leq n$, let $\psi_i$ be a set of $cn$ clauses picked uniformly at random from the set of clauses containing 1 literal from the set $X_i = \{ x_i^1 , \dots, x_i^m \}$ and 2 literals from $Y_i = \{ y_i^1 , \dots, y_i^n \}$. Further suppose that $m \leq \log_2 (n)$ and $c,l$ are any constants with $1<c<2$, $l < 1$. With high probability at least $ln$ of the $\Psi_i$ have $\cost(\Psi_i) \geq 2$.
\label{lem:nonconstant-psis}
\end{lem}

\begin{proof}
In Lemma~\ref{lem:psisfalse}, we saw that with high probability $\Psi_i = \exists Y_i \forall X_i \cdot \psi_i$ is false for every $1 \leq i \leq n$, so $\cost(\Psi_i)$ is defined for all $\Psi_i$.
By Lemma~\ref{lem:nonconststrat}, for each $\Psi_i$, the probability that $\cost(\Psi_i) \geq 2$ is  $1 - o(1)$. Using the Hoeffding bound on the sum of independent Bernoulli random variables, the probability that fewer than $ln$ of the $\Psi_i$ satisfy $\cost(\Psi_i) \geq 2$ is at most
$$  \exp \left( -2 {\left( 1 - l - o(1) \right) }^2 n \right) $$
which for sufficiently large $n$ can be upper bounded by
$$ \exp \left( -2 {\left( 1 - l' \right) }^2 n \right) $$
for some constant $l' < 1$. Thus with probability $1 - o(1)$ at least $ln$ of the $\Psi_i$ have cost at least~2.
\end{proof}

Lemma~\ref{lem:nonconstant-psis} shows that in the randomly generated QBF $Q(n,m,c) \equiv \bigvee_i \Psi_i$, for suitable values of $m$ and $c$, the $\Psi_i$ are all false and with high probability, a linear proportion of them have $\cost(\Psi_i) \geq 2$. With a slightly more careful choice of $m$, these two properties suffice to show a cost lower bound for $Q(n,m,c)$. Unfortunately, we cannot obtain $\cost(Q(n,m,c))$ simply by multiplying $\cost(\Psi_i)$ for each $ i \in [n]$, as responses on $\vars_\forall (\Psi_i)$ may now vary depending on the assignment of variables in some other $\Psi_j$. Instead, we use the fact that if $\cost(\Psi_i) \geq 2$, then for any response $\beta_i$ there is some existential assignment  for which $\beta_i$ is \emph{not} a winning response. Using these, for any response $\beta$ for $Q(n,m,c)$, we construct a large set of existential assignments for which $\beta$ is not a winning response.

\begin{lem}
\label{lem:rand-highcost}
  Let $1 < c < 2$ be a constant, and let $m \leq (1 - \epsilon) \log_2 (n)$ for some constant $\epsilon > 0$. With probability $1 - o(1)$, $Q(n,m,c)$ is false and $\cost (Q(n,m,c)) = 2^{\Omega({n}^\epsilon)}$.
\end{lem}

\begin{proof}
For sets $Y_i = \{ y_i^1 , \dots , y_i^n \}$, $X_i = \{ x_i^1 , \dots, x_i^m \}$, with $m \leq (1 - \epsilon) \log_2 (n)$, 
$$ Q(n,m,c) := \exists Y_1 \dots Y_n \forall X_1 \dots X_n \exists t_1 \dots t_n \cdot \bigwedge_{i=1}^n \bigwedge_{j=1}^{cn} \left( \neg t_i \vee C_i^j \right) \wedge \bigvee_{i=1}^n t_i$$
 where the clauses $C_i^j$ are chosen uniformly at random to contain two literals on variables in $Y_i$ and a literal on a variable of $X_i$.
For each $1 \leq i \leq n$, define $\Psi_i = \exists Y_i \forall X_i \cdot \bigwedge_{j=1}^{cn} C_i^j$. 
Let $0< l < 1$ be a constant. By Lemma~\ref{lem:nonconstant-psis}, with probability $1-o(1)$, all the $\Psi_i$ are false, and at least $ln$ of the $\Psi_i$ have $\cost(\Psi_i) \geq 2$. It therefore suffices to show that if all the $\Psi_i$ are false and $\cost(\Psi_i) \geq 2$ holds for at least $ln$ of the $\Psi_i$, then $Q(n,m,c)$ is false and $\cost(Q(n,m,c)) \geq 2^{\Omega(n^\epsilon)}$.

If each $\Psi_i$ is false, there is some winning strategy $S_i : \langle Y_i \rangle \to \langle X_i \rangle$ for each $i \in [n]$. Define $S: \langle Y_1 , \dots , Y_n \rangle \to \langle X_1 , \dots , X_n \rangle$ by $S(\alpha_1 , \dots, \alpha_n) = ( S_1 (\alpha_1) , \dots, S_n (\alpha_n))$. 
For any $\alpha \in \langle Y_1 , \dots, Y_n \rangle$, restricting $Q(n,m,c)$ by $\alpha$ and $S(\alpha)$ gives
$$Q(n,m,c)[\alpha][S(\alpha)] = \exists t_1  \dots t_n \cdot \bigwedge_{i=1}^n \bigwedge_{j=1}^{cn} \left( \neg t_i \vee C_i^j[\alpha|_{Y_i}][S_i(\alpha|_{Y_i})] \right) \wedge \bigvee_{i=1}^n t_i $$
but by definition of the strategies $S_i$, $C_i^j[\alpha_i][S_i(\alpha_i)] = \bot$ for some $1 \leq j \leq cn$, and so 
$$Q(n,m,c)[\alpha][S(\alpha)] = \exists t_1 \dots t_n \cdot \bigwedge_{i=1}^n \neg t_i \wedge \bigvee_{i=1}^n t_i$$ 
which is clearly unsatisfiable. Since $S$ is a winning $\forall$-strategy for $Q(n,m,c)$, $Q(n,m,c)$ is false if all the $\Psi_i$ are false for each $i \in [n]$.

It remains to show that $\cost(Q(n,m,c)) \geq 2^{\Omega(n^\epsilon)}$. We may assume that at least $ln$ of the $\Psi_i$ do not have constant winning $\forall$-strategies. Without loss of generality, we further assume that these are $\Psi_1 , \dots, \Psi_{ln}$, and that  all winning $\forall$-strategies for $Q(n,m,c)$ we consider assign the variables of $X_{ln+1}, \dots, X_n$ according to some constant winning $\forall$-strategy for $\Psi_{ln+1} , \dots , \Psi_n$. We therefore restrict our attention to strategies which are winning $\forall$-strategies for $\Psi_1 , \dots, \Psi_{ln}$.

Since $|X_i| \leq (1 - \epsilon) \log_2 (n)$, we can list the possible responses in each $\langle X_i \rangle$ as $\langle X_i \rangle = \{ \beta_1^i , \dots , \beta_N^i \}$, where $N = 2^m \leq  n^{(1- \epsilon)}$.

Let $B = \rng (S)$ for some winning $\forall$-strategy $S$ for $Q(n,m,c)$. To lower bound $\cost(Q(n,m,c))$, we need to show a lower bound on $|B|$. Given we assume $S$ is constant on $\Psi_{ln+1}, \dots , \Psi_n$, we can consider each  $\beta \in B$ as an assignment in $\langle X_1 , \dots, X_{ln} \rangle$, i.e. $B \subseteq \{ (\beta_{j_1}^i , \dots , \beta_{j_{ln}}^{ln}) : j_1 , \dots, j_{ln} \in [N]\}$. As $B$ is the image of a winning $\forall$-strategy, it contains a winning response $\beta$ for every assignment $\alpha \in \langle Y_1 , \dots , Y_{ln} \rangle$. In this case a winning response for $\alpha$ is some $\beta$ such that $\Psi_i[\alpha][\beta]$ is false for \emph{every} $1 \leq i \leq ln$.

For each $1 \leq i \leq ln$, $\Psi_i$ does not have a constant winning $\forall$-strategy. For any $\beta_j^i \in \langle X_i \rangle$, there is some assignment $\alpha_j^i \in \langle Y_i \rangle$ such that $\beta_j^i$ is not a winning response to $\alpha_j^i$ for $\Psi_i$. That is, for each $\beta_j^i$, there is some $\alpha_j^i$ such that $\Psi_i [\alpha_j^i][\beta_j^i] = \top$, else $\beta_j^i$ would define a constant winning $\forall$-strategy for $\Psi_i$ and $\cost(\Psi_i) = 1$. We now use these $\alpha_j^i$ to construct a multiset of existential assignments for which any response $\beta$ is only a winning response to a small subset.

Define the multiset $A$, containing elements of $\langle Y_1 , \dots, Y_{ln} \rangle$, as 
$$A = \left\{ (\alpha_{j_1}^1 , \dots , \alpha_{j_{ln}}^{ln} ) : (j_1 , \dots, j_{ln}) \in {[N]}^{ln} \right\}$$
Note that  $\alpha_j^i$ and $\alpha_{j'}^i$ need not be distinct for $j \neq j'$, so defining $A$ to be a multiset ensures $|A| = N^{ln}$. Given a response $\beta \in B$, we bound the size of the multiset 
$$A_\beta = \left\{ \alpha \in A : \Psi_i [\alpha][\beta] = \bot \mbox{ for all } 1 \leq i \leq ln \right\}$$
the set of all assignments in $A$ for which $\beta$ is a winning response, counted with their multiplicity in $A$.

For any assignment $\beta \in B$, we know $\beta = (\beta_{j_1}^1 , \dots, \beta_{j_{ln}}^{ln})$ for some $j_1 , \dots , j_{ln} \in [N]$. If $\beta$ is a winning response to $\alpha$, then $\Psi_i[\alpha|_{Y_i}][\beta|_{X_i}] = \bot$ for all $1 \leq i \leq ln$. Since $\beta|_{X_i} = \beta_{j_i}^i$, by the definition of $\alpha_{j_i}^i$, $\alpha|_{Y_i} \neq \alpha_{j_i}^i$ for all $1 \leq i \leq ln$, as $\Psi_i[\alpha_{j_i}^i][\beta_{j_i}^i] = \top$. The restriction $\alpha|_{Y_i} \neq \alpha_{j_i}^i$ restricts the set $A_\beta$ to
$$ A_\beta \subseteq \left\{ ( \alpha_{k_1}^1 , \dots, \alpha_{k_{ln}}^{ln}) :  (k_1 , \dots, k_{ln}) \in {[N]}^{ln} ~ , ~ j_i \neq k_i \mbox{ for all } 1 \leq i \leq ln \right\} $$
In particular, we see that $|A_\beta| \leq {(N-1)}^{ln}$.

For any $\alpha \in A$, let $\beta = S(\alpha) \in B$. Since $S$ is a winning $\forall$-strategy, $\beta$ is a winning response to $\alpha$, or equivalently $\alpha \in A_\beta$.
By definition, $A_\beta \subseteq A$ for each $\beta \in B$, so it is clear that $A = \bigcup_{\beta \in B} A_\beta$. Comparing the cardinalities of these sets gives $N^{ln} \leq \sum_{\beta \in B} |A_\beta| \leq |B| {(N-1)}^{ln}$, and so $|B| \geq { \left( \frac{N}{N-1} \right) }^{ln}$. For $N > 1$, this is a monotonically decreasing function in $N$, and $N \leq n^{(1 - \epsilon)}$, so for sufficiently large $n$,
$$ |B| \geq {\left( \frac{N}{N-1} \right) }^{ln} \geq{\left( \frac{n^{(1-\epsilon)}}{n^{(1-\epsilon)}-1} \right) }^{ln} = {\left( 1 + \frac{1}{n^{(1-\epsilon)}} \right) }^{ln} = { \left({\left( 1 + \frac{1}{n^{(1-\epsilon)}} \right) }^{n^{(1-\epsilon)}} \right) }^{l {n}^\epsilon} = 2^{\Omega (n^\epsilon)} $$
since for large $n$, ${\left( 1 + \frac{1}{n} \right) }^n \geq 2$.
We conclude that $|\rng(S)| \geq 2^{\Omega(n^\epsilon)}$ for any winning $\forall$-strategy $S$. There is only one block of universal variables in $Q(n,m,c)$, and so $$\cost(Q(n,m,c)) = \min \{ |\rng(S)| : \mbox{$S$ is a winning $\forall$-strategy for $Q(n,m,c)$} \} \geq 2^{\Omega(n^\epsilon)}$$

We have shown that if all the $\Psi_i$ are false, then $Q(n,m,c)$ is false, and further that if at least $ln$ of the $\Psi_i$ have no constant winning $\forall$-strategy, then $\cost(Q(n,m,c)) \geq 2^{\Omega(n^\epsilon)}$. Lemma~\ref{lem:nonconstant-psis} states that these conditions hold with probability $1 - o(1)$, and this completes the proof.
\end{proof}

Lemma~\ref{lem:rand-highcost} proves that, for the appropriate values of $m$ and $c$, the QBFs $Q(n,m,c)$ are false and have large cost with probability $1 - o(1)$. It is then a simple application of Size-Cost-Capacity and the capacity upper bounds shown in Section~\ref{sec:cap-bounds} to show lower bounds on $Q(n,m,c)$ with high probability.

\begin{thm}
\label{thm:random-hardness}
Let $1 < c < 2$ be a constant, and let $m \leq (1 - \epsilon) \log_2 (n)$ for some constant $\epsilon > 0$. With high probability, the randomly generated QBF $Q(n, m , c)$ is false, and any $\QURes$, $\CPred$ or $\PCRred$ refutation of $Q(n,m,c)$ requires size $2^{\Omega(n^\epsilon)}$.
\end{thm}

As previously, the greater capacity of lines in $\Fred$ does allow for short proofs of $Q(n,m,c)$ whenever it is false.
Refuting any individual false $\Psi_i$ is easy, even for $\QURes$. Applying $\forall$-reduction to each clause results in an unsatisfiable 2-SAT instance, which has a linear size resolution refutation. This immediately gives a short $\Fred$ proof for any false $Q(n,m,c)$, by deriving $\bigvee_{i=1}^n \Psi_i$, and then refuting each $\Psi_i$ in turn.

\section{Easy lower bounds for the formulas of Kleine B\"{u}ning et al.} %
\label{sec:kleinebuening}

We conclude with a new proof of the lower bounds on the prominent formulas of Kleine B\"{u}ning et al. \cite{BuningKF95} using Size-Cost-Capacity.

\begin{defi}[Kleine B\"{u}ning et al. \cite{BuningKF95}]
The formulas $\kappa(n)$ are defined as 
$$ \kappa (n) := \exists y_0 (\exists y_1 \exists y_1' \forall u_1)  \dots (\exists y_k \exists y_k' \forall u_k) \dots (\exists y_n \exists y_n' \forall u_n ) (\exists y_{n+1} \dots y_{n+n}) \cdot \bigwedge_{i=1}^{2n} C_i \wedge C_i'$$
where the matrix contains the clauses
$$ \begin{aligned}
C_0 &= \{ \neg y_0 \} ~&~ C_0' &= \{ y_0 , \neg y_1 , \neg y_1' \} \\
C_k &= \{ y_k , \neg u_k , \neg y_{k+1} \neg y_{k+1}' \} ~&~ C_k' &= \{ y_k', u_k, \neg y_{k+1} , \neg y_{k+1}' \} \\
C_n &= \{ y_n , \neg u_n , \neg y_{n+1} , \dots , \neg y_{n+n} \} ~&~ C_n' &= \{ y_n' , u_n , \neg y_{n+1} , \dots , \neg y_{n+n} \} \\
C_{n+t} &= \{ \neg u_t , y_{n+t} \} ~&~ C_{n+t}' &= \{ u_t , y_{n+t} \}
\end{aligned} $$
with $1 \leq k \leq n-1$ and $1 \leq t \leq n$.
\end{defi}

We also define the QBF $\lambda(n)$ constructed by adding the universal variables $v_k$ for each $1 \leq k \leq n$, quantified in the same block as $u_k$. The matrix of $\lambda(n)$ contains the clauses $D_i, D_i'$, where each $D_i, D_i'$ consists of the literals in the corresponding $C_i,C_i'$, but for each literal on some $u_k$, we add the matching literal on $v_k$. This is essentially `doubling' the variables $u_k$ with the matching variables $v_k$. The effect of this is to prevent any resolution steps being possible on universal variables before the variables can be $\forall$-reduced.

In \cite{BuningKF95,BeyersdorffCJ15}, it was shown that $\kappa(n)$ requires proofs of size $2^n$ in $\QRes$, which is $\QURes$ in which universal variables cannot be used as resolution pivots. This lower bound immediately transfers to the same lower bound for $\lambda(n)$ in $\QURes$ \cite{BalabanovWJ14}. As one of the first QBF lower bounds to be shown, these formulas have been the subject of much attention in the study of QBF proof complexity (for examples, see \cite{Egly16,BeyersdorffCJ15,BalabanovWJ14,LonsingES16}).

Showing a lower bound for $\kappa(n)$ in $\QRes$ is equivalent to showing a lower bound for $\lambda(n)$ in $\QURes$. It can be assumed in both proof systems that $\forall$-reductions are performed whenever possible, and so all clauses in the shortest $\QURes$ proof either contain matching literals on $u_k$ and $v_k$, or contain no literal on either of them. Any two such clauses cannot be used in a resolution step on a universal variable $u_k$, as the resulting clause would contain both $v_k$ and $\neg v_k$. All clauses derived from this clause will contain $v_k$ and $\neg v_k$, until a $\forall$-reduction reduces the clause to $\top$. The shortest $\QURes$ proof of $\lambda(n)$ therefore contains no resolution steps on universal pivots, and so the same steps can be used to produce a $\QRes$ proof of $\kappa(n)$.

We use Size-Cost-Capacity to prove a $\QURes$ lower bound for an even weaker QBF than $\lambda(n)$, which is obtained by quantifying all the variables $v_k$ in the rightmost universal block. This allows us to give a cost lower bound using this block, which in turn gives the proof size lower bound.

\begin{prop}
\label{prop:kb-cost}
The QBF $\lambda' (n) := \exists y_0 y_1 y_1' \forall u_1  \dots \exists y_n y_n' \forall u_n v_1 \dots v_n \exists y_{n+1} \dots y_{n+n} \cdot \bigwedge_{i=1}^{2n} D_i \wedge D_i'$ has cost $2^n$.
\end{prop}

\begin{proof}
We consider the response of any winning $\forall$-strategy to the $2^n$ distinct assignments in the set $A = \{ \alpha \in \langle \{ y_1, y_1' , \dots , y_n, y_n' \} \rangle : \alpha(y_k) \neq \alpha(y_k') \mbox{ for all } 1 \leq k \leq n \}$. Any assignment in $A$ forces a winning $\forall$-strategy $S$ to respond by setting $u_k = y_k'$. If not, then all clauses $C_i, C_i'$ for $i \leq k$ would be satisfied, and the further assignment $y_j = y_j' = 1$ for all $j > k$ would satisfy the matrix.

It remains to show that responding with $v_k = u_k$ is the only possible response to any $\alpha \in A$ for a winning $\forall$-strategy. We demonstrate this in the case of the assignment $\alpha$ where $\alpha(y_k) = 1$, $\alpha(y_k') = 0$ for all $1 \leq k \leq n$, but other assignments in $A$ are similar. Restricting by the assignment $\alpha$, as well as $\beta$, where $\beta(u_k) = \alpha(y_k')$ as shown above, the restricted matrix contains the clauses
$$
\begin{aligned}
D_n'|_{\alpha \cup \beta} &= \{ v_n , \neg y_{n+1}, \dots , \neg y_{n+n} \} &\\
D_{n+t}'|_{\alpha \cup \beta} &= \{ v_t , y_{n+t} \} &\mbox{ for each } 1 \leq t \leq n.
\end{aligned}
$$

If $S_n(\alpha)$ sets any $v_k = 1$, then the matrix is clearly satisfiable by setting $y_{n+k} = 0$, and $y_{n+j} = 1$ for all $j \neq k$. There is therefore a unique response on the $v_k$ for $S_n(\alpha)$, which is to set $v_k = u_k = y_k'$. It is clear that there is a similar unique response for any $\alpha \in A$. We conclude that $|\rng(S_n)| = 2^n$, whence $\cost(\lambda'(n)) = 2^n$.
\end{proof}

We therefore obtain the following hardness result, which was known for $\QURes$ \cite{BalabanovWJ14}, but also lifts to $\CPred$ and $\PCRred$.
\begin{cor}
\label{cor:kb-size}
Any $\QURes$ , $\CPred$ or $\PCRred$ proof of $\lambda(n)$ requires size $2^{\Omega(n)}$.
\end{cor}

\begin{proof}
The only difference between $\lambda'(n)$ and $\lambda(n)$ is the order in which the variables are quantified. As the variables $v_k$ are quantified further right in $\lambda'(n)$, any refutation of $\lambda(n)$ in any of these proof systems is also a refutation of $\lambda' (n)$. It is therefore sufficient to show a lower bound for refutations of $\lambda'(n)$ in $\QURes$, $\CPred$ or $\PCRred$, which is an immediate consequence of Proposition~\ref{prop:kb-cost} and the results of Section~\ref{sec:cap-bounds}.
\end{proof}

This lower bound for $\QURes$ also yields the lower bound on $\QRes$ proofs of $\kappa(n)$, previously shown in \cite{BuningKF95,BeyersdorffCJ15}. As well as providing a relatively simple proof of the hardness of these formulas, our technique also offers some insight as to why these formulas are hard. As the strategy for each variable $u_k$ is simple to compute in even very restricted models of computation, and the proof size lower bounds do not arise from propositional lower bounds, the lower bounds on $\kappa(n)$ and $\lambda(n)$ seemed to be something of an anomaly among QBF proof complexity lower bounds \cite{BeyersdorffHP17}. Here we have shown that the lower bound ultimately arises from the cost of the formulas, although this is slightly obfuscated by some rearrangement of the quantifier prefix.

\section{Conclusions and Open Problems} %
\label{sec:conclusions}

By formalising the conditions on $\PS$ in the construction of $\Pred$, we developed a new technique for proving QBF lower bounds in $\Pred$.
We demonstrated that the technique is not restricted to a few carefully constructed QBFs, but is in fact applicable to a large class of randomly generated formulas.
The primary appeal of the technique is its semantic nature.
We believe that lower bounds based on semantic properties of instances, as opposed to syntactic properties of proofs, work to further our understanding of the hardness phenomenon across the wider range of QBF proof systems.
We strongly suggest that Size-Cost-Capacity is applicable beyond $\Pred$, and future work will likely establish the hardness of random QBFs in even stronger systems (for example in the expansion based calculus $\IRC$ \cite{BeyersdorffCJ14}).

Recent work on reasons for hardness in QBF proof systems \cite{BeyersdorffHP17} proposed a new model that precisely characterises the notion of non-genuine hardness in $\Pred$.
In that model, the system $\SigmaPred$ utilises an $\PS$ oracle to derive any propositional implicant in a single step, whereby proof size can be measured simply in terms of universal reduction.
Hence, by counting universal reduction steps, the Size-Cost-Capacity Theorem in fact gives an absolute proof-size lower bound in $\SigmaPred$.
We note that our technique also applies to the so-called semantic derivations appearing in \cite{Krajicek02}, since the use of an $\PS$ oracle guarantees the simulation of any propositional semantic inference.

We also introduced an interesting new formula family, the equality formulas, arguably the simplest known hard QBFs.
When considering QBF proof complexity lower bounds, particularly in $\Pred$ systems, we must concern ourselves with formulas with at least a $\Sigma_3$ prefix, of which the equality formulas are one of the simplest examples.
If a QBF has a $\Sigma_2$ prefix, then it is true if and only if the existential parts of the clauses can all be satisfied, i.e. it is equivalent to a SAT problem. Similarly, a refutation of a QBF with a $\Pi_2$ prefix consists of a refutation of a subset of the existential clauses corresponding to a particular assignment to the universal variables. A $\Pi_3$ formula can also be regarded essentially as a SAT problem using similar reductions as for both $\Sigma_2$ and $\Pi_2$, so $\Sigma_3$ is the smallest prefix where we can expect to find genuine QBF lower bounds.

\subsubsection*{Future Perspectives}
In strong proof systems such as $\Fred$, superpolynomial proof size lower bounds can be completely characterised: they are either a propositional lower bound or a circuit lower bound \cite{BeyersdorffP16}.
All QBFs we have considered have no underlying propositional hardness, and winning $\forall$-strategies can be computed by small circuits, even in very restricted circuit classes.
As such, all these QBFs are easy for $\Fred$.

However, for weaker proof systems, such as $\QURes$ and the others we have discussed, propositional hardness and circuit lower bounds alone are not the complete picture (cf.~\cite{BeyersdorffHP17}).
In particular, the lower bounds we have shown using Size-Cost-Capacity do not fit into either class.
That this technique relies on capacity upper bounds which do not hold for strong proof systems leads us to suggest that we have identified a new reason for the hardness of QBFs in those proof systems where the above dichotomy does not hold.
We believe this represents significant progress towards a similar characterisation of lower bounds for these proof systems.

\subsubsection*{Open Problems}
Our work raises several interesting open problems regarding the applicability of our technique to further QBF systems.
Firstly, what is the status of Size-Cost-Capacity with respect to long-distance Q-resolution?
It is quite simple to construct short proofs of the equality formulas in $\LDQRes$, which demonstrates that the technique as presented in this paper will not lift there, but the underlying ideas could potentially be adapted.
The question applies also to expansion-based QBF systems.
Finally, the complexity of our random formulas outside of the $\Pred$ paradigm remains an interesting open problem.

\section*{Acknowledgments}
Research was supported by a grant from the John Templeton Foundation (grant no.\ 60842).

\newcommand{\etalchar}[1]{$^{#1}$}


\end{document}